\documentclass[oneside,british,a4paper]{amsart}
\usepackage[T1]{fontenc}
\usepackage{textcomp}
\usepackage[utf8]{inputenc}
\usepackage{parskip}
\usepackage{babel}
\usepackage{verbatim}
\usepackage{prettyref}
\usepackage{mathtools}
\usepackage{amstext}
\usepackage{amsthm}
\usepackage{amssymb}
\usepackage[a4paper]{geometry}
\usepackage[pdfusetitle,
 bookmarks=true,bookmarksnumbered=false,bookmarksopen=false,
 breaklinks=false,pdfborder={0 0 1},backref=false,colorlinks=false]
 {hyperref}

\makeatletter
\numberwithin{equation}{section}
\numberwithin{figure}{section}

\@ifundefined{date}{}{\date{}}
\usepackage{prettyref}
\usepackage{enumerate}
\allowdisplaybreaks

\newcommand*{\e}{\mathrm{e}}
\renewcommand{\i}{\mathrm{i}}
\newcommand*{\ii}{\mathrm{i}}

\newcommand{\R}{\mathbb{R}}

\newcommand{\C}{\mathbb{C}}
\newcommand{\dom}{\operatorname{dom}}

\renewcommand{\d}{\,\mathrm{d}}

\newcommand{\diag}{\operatorname{diag}}

\newcommand{\trace}{\operatorname{trace}}

\newcommand{\sinc}{\operatorname{sinc}}
\newcommand{\conj}{\operatorname{conj}}

\newrefformat{subsec}{Subsection \ref{#1}}
\newrefformat{prob}{Problem \ref{#1}}
\newrefformat{prop}{Proposition \ref{#1}}
\newrefformat{lem}{Lemma \ref{#1}}
\newrefformat{thm}{Theorem \ref{#1}}
\newrefformat{cor}{Corollary \ref{#1}}
\newrefformat{rem}{Remark \ref{#1}}
\newrefformat{exa}{Example \ref{#1}}
\newrefformat{sub}{Subsection \ref{#1}}
\newrefformat{eq}{(\ref{#1})}

\theoremstyle{definition}

\newrefformat{hyp}{Hypotheses \ref{#1}}

\makeatother

\theoremstyle{plain}
\newtheorem{thm}{\protect\theoremname}[section]
\theoremstyle{remark}
\newtheorem{rem}[thm]{\protect\remarkname}
\theoremstyle{definition}
\newtheorem*{defn*}{\protect\definitionname}
\theoremstyle{plain}
\newtheorem{prop}[thm]{\protect\propositionname}
\newtheorem{lem}[thm]{\protect\lemmaname}
\newtheorem{cor}[thm]{\protect\corollaryname}
\providecommand{\corollaryname}{Corollary}
\providecommand{\definitionname}{Definition}
\providecommand{\lemmaname}{Lemma}
\providecommand{\propositionname}{Proposition}
\providecommand{\remarkname}{Remark}
\providecommand{\theoremname}{Theorem}

\begin{document}
\title{An Elementary Perspective on the Dirac Equation and its Closest Relatives}

\author[R.~Picard]{Rainer Picard}
\address[R.P.]{TU Dresden\\
Fakult\"at f\"ur Mathematik \\
  Institut für Analysis\\
Dresden \\  Germany}
  \email{rainer.picard@magenta.de}

\author[S.~Trostorff]{Sascha Trostorff}
\address[S.T.]{CAU Kiel \\
Mathematisches Seminar \\
Arbeitsbereich Analysis \\
Kiel \\
  Germany}
\email{trostorff@math.uni-kiel.de}

\author[M.~Waurick]{Marcus Waurick}
\address[M.W.]{TU Bergakademie Freiberg \\
Fakult\"at f\"ur Mathematik und Informatik \\
  Institut f\"ur Angewandte Analysis 
\\ Freiberg \\
  Germany}
\email{marcus.waurick@math.tu-freiberg.de}

\begin{abstract} We consider the Dirac equation in a real setting to show
its close relationship to Maxwell's equation. In addition to analysing
the classical Dirac mass term, we also inspect the alternative Majorana
mass term and discuss their quaternionic structure and its consequences.
\end{abstract}

\maketitle

\subjclass{\textbf{MSC 2020}} {Primary 35Q61, Secondary 35Q41, 81V99, 46S05, 47S05, 15B33}

\dedicatory{}

\keywords{\textbf{Keywords} Dirac operator, extended Maxwell system, Majorana particle, quaternions}

\setcounter{section}{-1}

\section{Introduction}

When Paul Dirac, \cite{Dirac1928QTElectron,Dirac1928QuantumTheoryElectronPartII},
found his factorisation of the Klein--Gordon operator leading to
the celebrated Dirac equation, it was generally taken for granted
that this is the smallest factorisation of the Klein--Gordon operator.
However, Majorana's 1937, \cite{Majorana1937TeoriaSimmetrica}, alternative
proposal showed that actually there is a factorisation of the Klein--Gordon
operator of half the size (which seemingly Majorana did not fully
realise). This rather spectacular finding was largely dismissed as
irrelevant, since by this time the development of quantum mechanics
based on Dirac's approach had already been developed too far to be
questioned easily. In particular, it was held that complex linearity
is of the essence in this setting, whereas the Majorana operator would
reduce to real coefficients. Still, from a mathematical perspective
this is an interesting fact, which we want to explore more deeply
in this expository note. The curious connection between the Dirac
equation and the extended Maxwell system is another ingredient, which
fosters the idea that maybe a completely real approach might lead
to some new insights. After all, if linear isometric transformations
are thought to yield just another description of ``the same physics'',
why should we for a particular area of physics only allow such transformations,
if they additionally are $\mathbb{C}$-linear. Even if we have complex
elements in a Hilbert space there is no need to consider it as a complex
Hilbert space (multiplication by $\ii$ is then just an important
skew-selfadjoint operator). We can simply take the real part of the
sesquilinear complex inner product to obtain a real inner product.
In this view we have that
\begin{align*}
\mathbb{C} & \to\mathbb{R}^{2}\\
x+\ii y & \mapsto\left(\begin{array}{c}
x\\
y
\end{array}\right)
\end{align*}
is a (real) unitary mapping. 

For example, let us take
\[
\ii\:\left(\begin{array}{cc}
0 & -\partial\\
\partial & 0
\end{array}\right)
\]
as an operator describing some imagined physical situation via some
equation in $L^{2}\left(\mathbb{R},\mathbb{C}^{2}\right)$, $\partial$
denoting differentiation. Separating real and imaginary part this
amounts to considering
\[
\left(\begin{array}{cc}
\left(\begin{array}{cc}
0 & 0\\
0 & 0
\end{array}\right) & \left(\begin{array}{cc}
0 & \partial\\
-\partial & 0
\end{array}\right)\\
\left(\begin{array}{cc}
0 & -\partial\\
\partial & 0
\end{array}\right) & \left(\begin{array}{cc}
0 & 0\\
0 & 0
\end{array}\right)
\end{array}\right)
\]
 in $L^{2}\left(\mathbb{R},\mathbb{R}^{4}\right).$ Note that the
complex structure is now hidden in the fact that
\[
\left(\begin{array}{cc}
\left(\begin{array}{cc}
0 & 0\\
0 & 0
\end{array}\right) & \left(\begin{array}{cc}
0 & \partial\\
-\partial & 0
\end{array}\right)\\
\left(\begin{array}{cc}
0 & -\partial\\
\partial & 0
\end{array}\right) & \left(\begin{array}{cc}
0 & 0\\
0 & 0
\end{array}\right)
\end{array}\right)\text{ commutes with }\left(\begin{array}{cc}
\left(\begin{array}{cc}
0 & -1\\
1 & 0
\end{array}\right) & \left(\begin{array}{cc}
0 & 0\\
0 & 0
\end{array}\right)\\
\left(\begin{array}{cc}
0 & 0\\
0 & 0
\end{array}\right) & \left(\begin{array}{cc}
0 & -1\\
1 & 0
\end{array}\right)
\end{array}\right).
\]
The real isometry $U=\left(\begin{array}{cc}
\left(\begin{array}{cc}
1 & 0\\
0 & 1
\end{array}\right) & \left(\begin{array}{cc}
0 & 0\\
0 & 0
\end{array}\right)\\
\left(\begin{array}{cc}
0 & 0\\
0 & 0
\end{array}\right) & R
\end{array}\right)$ ($R$ orthogonal matrix) transforms this to an operator
\[
U\left(\begin{array}{cc}
\left(\begin{array}{cc}
0 & 0\\
0 & 0
\end{array}\right) & \left(\begin{array}{cc}
0 & \partial\\
-\partial & 0
\end{array}\right)\\
\left(\begin{array}{cc}
0 & -\partial\\
\partial & 0
\end{array}\right) & \left(\begin{array}{cc}
0 & 0\\
0 & 0
\end{array}\right)
\end{array}\right)U^{*}=\left(\begin{array}{cc}
\left(\begin{array}{cc}
0 & 0\\
0 & 0
\end{array}\right) & \left(\begin{array}{cc}
0 & \partial\\
-\partial & 0
\end{array}\right)R^{*}\\
R\left(\begin{array}{cc}
0 & -\partial\\
\partial & 0
\end{array}\right) & \left(\begin{array}{cc}
0 & 0\\
0 & 0
\end{array}\right)
\end{array}\right),
\]
which is in general not commuting with $\left(\begin{array}{cc}
\left(\begin{array}{cc}
0 & -1\\
1 & 0
\end{array}\right) & \left(\begin{array}{cc}
0 & 0\\
0 & 0
\end{array}\right)\\
\left(\begin{array}{cc}
0 & 0\\
0 & 0
\end{array}\right) & \left(\begin{array}{cc}
0 & -1\\
1 & 0
\end{array}\right)
\end{array}\right)$ anymore. Indeed, we would need to have $R=\left(\begin{array}{cc}
1 & 0\\
0 & 1
\end{array}\right)$ to maintain commutativity with $\left(\begin{array}{cc}
\left(\begin{array}{cc}
0 & -1\\
1 & 0
\end{array}\right) & \left(\begin{array}{cc}
0 & 0\\
0 & 0
\end{array}\right)\\
\left(\begin{array}{cc}
0 & 0\\
0 & 0
\end{array}\right) & \left(\begin{array}{cc}
0 & -1\\
1 & 0
\end{array}\right)
\end{array}\right)$.

So, looking at the situation from a completely real perspective will
be more flexible and this will indeed also be the starting point of
our exposition. As a convenient tool of presentation, we will make
extensive use of the intuitive concept of block matrices and block
operator matrices. This will in particular assist to avoid tedious
calculations with special bases. For a more in-depth theoretical framework,
we refer, for example to \cite{Picard2009,PicardMcGhee2011,Primer}.
To streamline the presentation further we have done away with the
issue of boundary conditions by simply taking $\mathbb{R}^{3}$ as
the underlying spatial domain and sticking with vector analytic language
(rather than differential forms). Boundary conditions are a separate
issue, which can be dealt with in a canonical way via inspecting associated
deficiency spaces (starting with the minimal operator, compare e.g.~\cite{PT2024}). 

Much of this is surely not really new. However, it is sadly not at
the forefront of discussion, but rather buried under layers of rather
complex constructions, which may easily serve as a possible distraction
from more basic and more accessible considerations.

Our plan is to start with the classical Heaviside--Gibbs formulation
of Maxwell's equations. Then we proceed to an extended Maxwell system,
which shows quaternionic features and may be considered as the Maxwell
system that Maxwell might have intended. From this we recover the
massless Dirac equation, here deliberately written in (real) vector
analytic language to keep matters accessible\footnote{Transcription in terms of exterior calculus of differential forms
is surely a preferable way to pursue this deeper. That the Dirac system
is actually a system acting on differential forms or alternating tensors,
is a long standing suspicion and was already discussed early on, see
\cite{IwanenkoLandau1928,Kaehler1961DiracGleichung,Kaehler1962InnererDifferentialkalkuelHamburg}.}. To also include a non-trivial mass term some basic algebraic ideas
are needed and provided in section \eqref{sec:Hilbert-Algebras-of}.
The mass term proposed by Dirac also is quaternionic in nature. We
then choose to consider the Weyl picture of the Dirac\footnote{\label{fn:quaternionic}The Dirac operator is in itself of quaternionic
operator structure in as much as its spatial differential operator
part is of the block operator form $\left(\begin{array}{cc}
A & -B^{*}\\
B & A^{*}
\end{array}\right)$ with commuting (skew-selfadjoint) entries.} equation, which shows the phenomenon of decoupling for vanishing
mass. The mass term proposed by Majorana goes a step further by showing
that the full Majorana--Dirac equation decouples fully in two equal
parts. This then leads to consider just one part as a canonical suggestion
for spin $0$ system that factorizes the Klein--Gordon operator
on a real $4\times4$ matrix basis. There are of course alternatives
to obtain first order evolution equations capturing the spin $0$
case. We merely mention the Feshbach--Villars theory, \cite{FeshbachVillars1958ElementaryRelativistic},
which is real size $4\times4$ but is a mixed order (Schrödinger type)
system and the Duffin--Kemmer--Petiau theory, \cite{Kemmer1939ParticleAspectMesonTheory},
where an intricate construction for a complex $5\times5$ (i.e., $10\times10$
in real terms) system operator is considered. Based on this observation
about the Majorana\footnote{The spatial Majorana operator is also in itself of quaternionic operator
structure in the sense of footnote \ref{fn:quaternionic}.} operator, we reconstruct Majoranas idea of higher fractional spin
systems and indeed extend it to all spins.

\section{The Extended Maxwell Operator}

To begin with, let us recall the classical Gibbs--Heaviside type
Maxwell equation
\[
\left(\partial_{0}+\mathcal{A}_{\mathrm{Max}}\right)\left(\begin{array}{c}
E\\
H
\end{array}\right)=\left(\begin{array}{c}
-J\\
0_{3\times1}
\end{array}\right)
\]
with
\[
\mathcal{A}_{\mathrm{Max}}\coloneqq\left(\begin{array}{cc}
0_{3\times3} & -{\nabla}\times\\
{\nabla}\times & 0_{3\times3}
\end{array}\right),
\]
where we denote by $\nabla\times$ the classical rotation of a vector-field.
In $L^{2}\left(\mathbb{R}^{3},\mathbb{R}^{3}\right)\oplus L^{2}\left(\mathbb{R}^{3},\mathbb{R}^{3}\right)$
the operator $\mathcal{A}_{\mathrm{Max}}$ with its natural maximal
domain; that is,
\[
\dom(\mathcal{A}_{\mathrm{Max}})=\{(E,H)\in L^{2}(\R^{3},\R^{3})^{2}\,;\,\nabla\times E,\nabla\times H\in L^{2}(\R^{3},\R^{3})\},
\]
where the rotation is taken in the distributional sense, turns out
to be skew-selfadjoint (\cite{Primer}). The commonly included divergence
equations are actually consequences of the above system. 

As a curiosity we note that Maxwell's equation in the above simplified
form can also be written in complex notation as
\[
\left(\partial_{0}+\ii\nabla\times\right)\left(E+\ii H\right)=-J,
\]
known as the Riemann--Silberstein representation of Maxwell\textquoteright s
equations, attributed historically to Silberstein \cite{Silberstein1907},
see also for a more recent discussion of associated issues \cite{Bialynicki-Birula1996}.

The so-called extended Maxwell equation, see \cite{Pi84,TV07}, allows
for the divergence equations to be included explicitly, provided that
the data are suitably matched. In contrast to the classical system,
which is of size $6\times6$, the extended Maxwell system is of size
$8\times8$ and the corresponding equation is of the form
\[
\left(\partial_{0}+\mathcal{A}_{\mathrm{eMax}}\right)\left(\begin{array}{c}
\left(\begin{array}{c}
\varphi\\
E
\end{array}\right)\\
\left(\begin{array}{c}
\psi\\
H
\end{array}\right)
\end{array}\right)=\left(\begin{array}{c}
\left(\begin{array}{c}
0_{1\times1}\\
-J
\end{array}\right)\\
\left(\begin{array}{c}
q\\
0_{3\times1}
\end{array}\right)
\end{array}\right)
\]
\[
\mathcal{A}_{\mathrm{eMax}}=\left(\begin{array}{cc}
\left(\begin{array}{cc}
0_{1\times1} & 0_{1\times3}\\
0_{3\times1} & 0_{3\times3}
\end{array}\right) & -\left(\begin{array}{cc}
0_{1\times1} & {\nabla}^{\top}\\
-\nabla & {\nabla}\times
\end{array}\right)\\
\left(\begin{array}{cc}
0_{1\times1} & {\nabla}^{\top}\\
-\nabla & {\nabla}\times
\end{array}\right) & \left(\begin{array}{cc}
0_{1\times1} & 0_{1\times3}\\
0_{3\times1} & 0_{3\times3}
\end{array}\right)
\end{array}\right)
\]
Now, the operator $\mathcal{A}_{\mathrm{eMax}}$ with its natural
maximal domain is again skew-selfadjoint, but now in $L^{2}\left(\mathbb{R}^{3},\mathbb{R}^{4}\right)\oplus L^{2}\left(\mathbb{R}^{3},\mathbb{R}^{4}\right)$,
where we employed the usual custom to denote the vector-analytical
operators $\operatorname{grad},\ \operatorname{curl},\ \operatorname{div}$
by $\nabla,\ \nabla\times,\nabla^{\top}$, respectively.
\begin{rem}
Due to the choice of the right-hand side with
\begin{equation}
\partial_{0}q+\nabla^{\top}J=0\label{eq:charge-conservation}
\end{equation}
it follows that
\[
\varphi=\psi=0,
\]
which is a core requirement for electrodynamics. In this way, Maxwell's
equations can be recovered from the extended Maxwell system. Conversely,
the requirement of the scalar parts to vanish, requires the relation
\eqref{eq:charge-conservation} between charge and current, see Remark
\ref{rem:As-another-historical} below.
\end{rem}

\begin{defn*}
We introduce the mapping
\begin{align*}
Q:\mathbb{R}\times\mathbb{R}^{3} & \to\mathbb{R}^{4\times4}\\
\left(r,e\right) & \mapsto Q\left(r,e\right)\coloneqq r1_{4\times4}+Q\left(0,e\right)\coloneqq\left(\begin{array}{cc}
r & e^{\top}\\
-e & r+e\times
\end{array}\right)\coloneqq\left(\begin{array}{cc}
r & e^{\top}\\
-e & r1_{3\times3}+e\times
\end{array}\right),
\end{align*}
where 
\[
e\times\coloneqq\left(\begin{array}{ccc}
0 & -e_{3} & e_{2}\\
e_{3} & 0 & -e_{1}\\
-e_{2} & e_{1} & 0
\end{array}\right)\quad(e\in\R^{3}).
\]
\end{defn*}
\begin{rem}
Using this notation, we get that formally
\[
Q\left(0,\nabla\right)=\left(\begin{array}{cc}
0 & {\nabla}^{\top}\\
-\nabla & {\nabla}\times
\end{array}\right)
\]
yielding
\[
\left(\begin{array}{cc}
Q\left(0,0\right) & -Q\left(0,\nabla\right)\\
Q\left(0,\nabla\right) & Q\left(0,0\right)
\end{array}\right)=\mathcal{A}_{\mathrm{eMax}}.
\]
Writing for simplicity just $\left(\begin{array}{cc}
0 & -Q\left(0,\nabla\right)\\
Q\left(0,\nabla\right) & 0
\end{array}\right)$ for $\mathcal{A}_{\mathrm{eMax}}$, we find 
\[
\mathcal{A}^{2}_{\mathrm{eMax}}=\left(\begin{array}{cc}
-Q\left(0,\nabla\right)^{2} & 0\\
0 & -Q\left(0,\nabla\right)^{2}
\end{array}\right).
\]
Moreover, 
\[
Q\left(0,\nabla\right)^{2}=\left(\begin{array}{cc}
-\nabla^{\top}\nabla & {\nabla}^{\top}(\nabla\times)\\
-(\nabla\times)\nabla & -\nabla\nabla^{\top}+({\nabla}\times)(\nabla\times)
\end{array}\right)=-\Delta\left(\begin{array}{cc}
1 & 0_{1\times3}\\
0_{3\times1} & 1_{\mathbb{R}^{3}}
\end{array}\right).
\]
In other words, $\mathcal{A}^{2}_{\mathrm{eMax}}$ is a factorisation
of $\Delta$, which is what Dirac was looking for by a different,
more complicated approach.
\end{rem}

\begin{rem}
\label{rem:As-another-historical}As another historical note, we remark
here that Cabibbo and Ferrari actually suggested, referring to an
idea of Dirac, \cite{Dirac1931Monopole}, although not in our perspective,
to construct the potential\footnote{Gauging appears here if we allow additional degrees of freedom by
modifying the right-hand side $\left(\begin{array}{c}
\left(\begin{array}{c}
0\\
E
\end{array}\right)\\
\left(\begin{array}{c}
0\\
H
\end{array}\right)
\end{array}\right)$ to $\left(\begin{array}{c}
\left(\begin{array}{c}
\psi\\
E
\end{array}\right)\\
\left(\begin{array}{c}
\varphi\\
H
\end{array}\right)
\end{array}\right)$, thus, in contrast to electrodynamics, allowing for scalar parts.
In the process of reconstructing $\left(\begin{array}{c}
\left(\begin{array}{c}
0\\
E
\end{array}\right)\\
\left(\begin{array}{c}
0\\
H
\end{array}\right)
\end{array}\right)$ from the potential $\Pi$ one simply ignores these scalar parts.} $\Pi$ to the electromagnetic field more symmetrically as a solution
of
\[
\left(\partial_{0}-\mathcal{A}_{\mathrm{eMax}}\right)\Pi=\left(\begin{array}{c}
\left(\begin{array}{c}
0\\
E
\end{array}\right)\\
\left(\begin{array}{c}
0\\
H
\end{array}\right)
\end{array}\right)
\]
to allow for magnetic monopoles, \cite{CabibboFerrari1962Monopoles}.
If magnetic monopoles are excluded the potential $\Pi$ is a solution
of the same equation but now $\Pi$ turns out to be of the form
\[
\Pi=\left(\begin{array}{c}
\left(\begin{array}{c}
0\\
A
\end{array}\right)\\
\left(\begin{array}{c}
A_{0}\\
0
\end{array}\right)
\end{array}\right),
\]
which is the usual concept of a potential\footnote{In the light of this observation, we note that the Proca equation
is just
\[
\left(\left(\partial_{0}+\mathcal{A}_{\mathrm{eMax}}\right)+m^{2}\left(\partial_{0}-\mathcal{A}_{\mathrm{eMax}}\right)^{-1}\right)\left(\begin{array}{c}
\left(\begin{array}{c}
0\\
E
\end{array}\right)\\
\left(\begin{array}{c}
0\\
H
\end{array}\right)
\end{array}\right)=\left(\begin{array}{c}
\left(\begin{array}{c}
0\\
-J
\end{array}\right)\\
\left(\begin{array}{c}
q\\
0
\end{array}\right)
\end{array}\right).
\]
After applying $\left(\partial_{0}-\mathcal{A}_{\mathrm{eMax}}\right)$
to the Proca operator $\left(\left(\partial_{0}+\mathcal{A}_{\mathrm{eMax}}\right)+m^{2}\left(\partial_{0}-\mathcal{A}_{\mathrm{eMax}}\right)^{-1}\right)$
we obtain four block Klein--Gordon operators
\[
\left(\partial_{0}-\mathcal{A}_{\mathrm{eMax}}\right)\left(\left(\partial_{0}+\mathcal{A}_{\mathrm{eMax}}\right)+m^{2}\left(\partial_{0}-\mathcal{A}_{\mathrm{eMax}}\right)^{-1}\right)=\partial^{2}_{0}-\Delta+m^{2}.
\]
Since, as the notation suggests, we maintain the requirement that
the scalar parts are zero, we get that only the two block Klein--Gordon
operators, which are acting on $E,H$, remain relevant.} in electrodynamics, compare \cite{PTW2017}. Of course, factoring
$\Delta$ is the whole point of using potentials in the first place.
It was thought to be advantageous to reduce electrodynamics to the
wave equation. Indeed,
\[
\left(\partial^{2}_{0}-\Delta\right)\Pi=\left(\partial_{0}+\mathcal{A}_{\mathrm{eMax}}\right)\left(\partial_{0}-\mathcal{A}_{\mathrm{eMax}}\right)\Pi=\left(\partial_{0}+\mathcal{A}_{\mathrm{eMax}}\right)\left(\begin{array}{c}
\left(\begin{array}{c}
0\\
E
\end{array}\right)\\
\left(\begin{array}{c}
0\\
H
\end{array}\right)
\end{array}\right)=\left(\begin{array}{c}
\left(\begin{array}{c}
0\\
-J
\end{array}\right)\\
\left(\begin{array}{c}
q\\
0
\end{array}\right)
\end{array}\right),
\]
and so finding $\Pi$ as the solution of the wave equation, one could
recover the electromagnetic field from $\left(\partial_{0}-\mathcal{A}_{\mathrm{eMax}}\right)\Pi$.
The same mechanism yields
\begin{align*}
\left(\partial^{2}_{0}-\Delta\right)\left(\begin{array}{c}
\left(\begin{array}{c}
\psi\\
E
\end{array}\right)\\
\left(\begin{array}{c}
\varphi\\
H
\end{array}\right)
\end{array}\right) & =\left(\partial_{0}-\mathcal{A}_{\mathrm{eMax}}\right)\left(\partial_{0}+\mathcal{A}_{\mathrm{eMax}}\right)\left(\begin{array}{c}
\left(\begin{array}{c}
\psi\\
E
\end{array}\right)\\
\left(\begin{array}{c}
\varphi\\
H
\end{array}\right)
\end{array}\right)\\
 & =\left(\partial_{0}-\mathcal{A}_{\mathrm{eMax}}\right)\left(\begin{array}{c}
\left(\begin{array}{c}
0\\
-J
\end{array}\right)\\
\left(\begin{array}{c}
q\\
0
\end{array}\right)
\end{array}\right)\\
 & =\left(\begin{array}{c}
\left(\begin{array}{c}
0\\
-\partial_{0}J-\nabla q
\end{array}\right)\\
\left(\begin{array}{c}
\partial_{0}q+\nabla^{\top}J\\
0
\end{array}\right)
\end{array}\right),
\end{align*}
and, by unique solvability of the source free wave equation (assuming
a proper setting), we read off that $\psi=0$ as well as $\varphi=0$,
as is required in electrodynamics, only if 
\[
\partial_{0}q+\nabla^{\top}J=0.
\]
This is, of course, the reason for the charge conservation law \eqref{eq:charge-conservation}.
\end{rem}

Allowing in the extended Maxwell system arbitrary right-hand sides,
we have indeed to account for non-vanishing scalar parts in the solution
(as in the potential construction). Somewhat miraculously, we then
get the Dirac equation 
\begin{equation}
\left(\partial_{0}+\mathcal{A}_{\mathrm{eMax}}\right)\Psi=F\label{eq:Max-m=00003D0}
\end{equation}
for spin $\frac{1}{2}$ particles with no mass, as we show in the
following section.

\section{The Dirac Equation for Massless Particles}

The Dirac equation is usually written in complex notation. We introduce
the following mapping.
\begin{defn*}
Set
\begin{align*}
\Phi_{\mathbb{C}}:\mathbb{R}^{4} & \to\mathbb{C}^{2}\\
\left(\begin{array}{c}
E_{0}\\
E_{1}\\
E_{2}\\
E_{3}
\end{array}\right) & \mapsto\left(\begin{array}{c}
E_{0}+\ii E_{1}\\
E_{2}+\ii E_{3}
\end{array}\right).
\end{align*}
\end{defn*}
\begin{rem}
It is straight forward to check that $\Phi_{\C}$ is a unitary operator.
\end{rem}

We will not distinguish between $\Phi_{\C}$ and the mapping acting
on $(\R^{4})^{n}$ as $\Phi_{\C}$ on each of the $n$ coordinates.
Applying this unitary transformation to \prettyref{eq:Max-m=00003D0}
we obtain
\[
\left(\partial_{0}+\Phi_{\mathbb{C}}\mathcal{A}_{\mathrm{eMax}}\Phi^{*}_{\mathbb{C}}\right)\Phi_{\mathbb{C}}\Psi=\Phi_{\mathbb{C}}F,
\]
now holding in $L^{2}\left(\mathbb{R}^{3},\mathbb{C}^{2}\right)\oplus L^{2}\left(\mathbb{R}^{3},\mathbb{C}^{2}\right)$,
where we identified
\[
\left(\begin{array}{c}
E_{0}\\
E
\end{array}\right)=\left(\begin{array}{c}
E_{0}\\
E_{1}\\
E_{2}\\
E_{3}
\end{array}\right)=\left(\begin{array}{c}
\left(\begin{array}{c}
E_{0}\\
E_{1}
\end{array}\right)\\
\left(\begin{array}{c}
E_{2}\\
E_{3}
\end{array}\right)
\end{array}\right).
\]
Since 
\[
\left(\partial_{0}+\Phi_{\mathbb{C}}\mathcal{A}_{\mathrm{eMax}}\Phi^{*}_{\mathbb{C}}\right)=\Phi_{\mathbb{C}}\left(\partial_{0}+\mathcal{A}_{\mathrm{eMax}}\right)\Phi^{*}_{\mathbb{C}},
\]
we have indeed what we shall call unitary congruence between these
two system operators, which is considered to describe ``the same
physics''. We will use the following notation:
\begin{defn*}
For two operators $A$ on $H_{0}$ and $B$ on $H_{1}$ for two Hilbert
spaces $H_{0}$ and $H_{1}$ we write 
\[
A\rightleftharpoons B
\]
if $A$ and $B$ are \textbf{unitarily congruent}\footnote{We shall give preference to this wording, which is commonly called
`unitarily equivalent'. Equivalence of $A,B$, however, is already
established in dealing with linear equations as meaning 
\[
A=UBV
\]
for bijections $U,V$. If $U=V^{*}$ then $A,B$ are called congruent.
If $U=V^{-1}$ then $A,B$ are called similar. Since we want $U,V$
to be unitary, which would appropriately be called then `unitarily
equivalent'. This is, however, not intended, but rather the latter
two, which do indeed coincide for the unitary case $U=V^{*}=V^{-1}$
and are correctly labelled as `unitarily congruent' or `unitarily
similar'. }; i.e., if there exists a unitary mapping $U:H_{0}\to H_{1}$ such
that 
\[
A=U^{\ast}BU.
\]
\end{defn*}
\begin{rem}
\begin{enumerate}
\item The rationale behind this ``the same physics'' perspective is that
since no instruments are known to actually measure a closed, densely
defined, linear operator $A:\dom\left(A\right)\subseteq H\to H$,
one has to live with the real numbers produced by the values of the
bilinear form 
\begin{align*}
H\times\dom\left(A\right)\subseteq H\times H & \to\mathbb{R}\\
\left(\psi,\phi\right) & \mapsto\left\langle \psi|A\varphi\right\rangle _{H}
\end{align*}
associated with the operator. Indeed, it is often preferred to resort
to the quadratic form\footnote{Due to the unnecessary use of complex Hilbert spaces looking at quadratic
forms is the only way to guarantee real quantities for selfadjoint
operators. Here we consider all Hilbert spaces to be real, thus bypassing
this problem altogether.}, which on first sight only seems to work if $A$ is selfadjoint.
We may, however, consider a bilinear form as a quadratic form on $\dom\left(A\right)\times\dom\left(A^{*}\right)\subseteq H\times H$.
Noting that then
\begin{align*}
\left\langle \psi|A\varphi\right\rangle _{H} & =\frac{1}{2}\left(\left\langle \psi|A\varphi\right\rangle _{H}+\left\langle A^{*}\psi|\varphi\right\rangle _{H}\right),\\
 & =\frac{1}{2}\left(\left\langle \psi|A\varphi\right\rangle _{H}+\left\langle \varphi|A^{*}\psi\right\rangle _{H}\right),\\
 & =\frac{1}{2}\left\langle \left(\begin{array}{c}
\varphi\\
\psi
\end{array}\right)\Big|\left(\begin{array}{cc}
0 & A^{*}\\
A & 0
\end{array}\right)\left(\begin{array}{c}
\varphi\\
\psi
\end{array}\right)\right\rangle _{H\oplus H},
\end{align*}
we see that this change of perspective yields the quadratic form associated
with the now selfadjoint operator $\left(\begin{array}{cc}
0 & A^{*}\\
A & 0
\end{array}\right)$, which can be recovered from this quadratic form. 
\item If we choose not to identify $H$ with its dual, we can consider the
bilinear form 
\begin{align*}
H^{\prime}\times\dom\left(A\right)\subseteq H^{\prime}\times H & \to\mathbb{R},\\
\left(\psi,\phi\right) & \mapsto\psi\left(A\varphi\right),
\end{align*}
instead. From this we can reconstruct $A$, but in the same way also
$A^{\prime}$ with $A^{\prime}:\dom\left(A^{\prime}\right)\subseteq H^{\prime}\to H^{\prime}$
denoting the dual operator of $A$. With the definition of the dual
operator we have
\[
\psi\left(A\varphi\right)=\left(A^{\prime}\psi\right)\left(\varphi\right)
\]
and recalling that in the Hilbert space case we identify $H^{\prime\prime}=H$,
we calculate analogously
\begin{align*}
\psi\left(A\varphi\right) & =\frac{1}{2}\left(\psi\left(A\varphi\right)+\left(A^{\prime}\psi\right)\left(\varphi\right)\right)\\
 & =\frac{1}{2}\left(\psi\left(A\varphi\right)+\varphi\left(A^{\prime}\psi\right)\right)\\
 & =\frac{1}{2}\left(\begin{array}{c}
\varphi\\
\psi
\end{array}\right)\left(\left(\begin{array}{cc}
0 & A^{\prime}\\
A & 0
\end{array}\right)\left(\begin{array}{c}
\varphi\\
\psi
\end{array}\right)\right)
\end{align*}
with
\[
\left(\begin{array}{cc}
0 & A^{\prime}\\
A & 0
\end{array}\right):\dom\left(A\right)\oplus\dom\left(A^{\prime}\right)\subseteq H\oplus H^{\prime}\to H^{\prime}\oplus H.
\]
The point being that we transitioned to the selfdual replacement $\left(\begin{array}{cc}
0 & A^{\prime}\\
A & 0
\end{array}\right):\dom\left(A\right)\oplus\dom\left(A^{\prime}\right)\subseteq H\oplus H^{\prime}\to H^{\prime}\oplus H$ for $\left(\begin{array}{cc}
0 & A^{*}\\
A & 0
\end{array}\right)$ instead. Then we can use the quadratic form
\begin{align*}
\dom\left(A\right)\oplus\dom\left(A^{\prime}\right)\subseteq H\oplus H^{\prime} & \to\mathbb{R},\\
\Phi & \mapsto\frac{1}{2}\;\Phi\left(\left(\begin{array}{cc}
0 & A^{\prime}\\
A & 0
\end{array}\right)\Phi\right),
\end{align*}
from which we can also reconstruct the operators $A$ and $A^{\prime}$,
analogously to the selfadjoint case, by simply differentiating this
quadratic functional. Indeed, finding stationary points of the action
functional
\[
\left(\Phi\mapsto\frac{1}{2}\;\Phi\left(\left(\begin{array}{cc}
0 & A^{\prime}\\
A & 0
\end{array}\right)\Phi\right)-F\left(\Phi\right)\right)
\]
amounts to simply setting the Fréchet derivative $\left(\Phi\mapsto\frac{1}{2}\;\Phi\left(\left(\begin{array}{cc}
0 & A^{\prime}\\
A & 0
\end{array}\right)\Phi\right)-F\left(\Phi\right)\right)^{\prime}$ of this functional equal to zero yielding the so-called Euler equation
\[
\left(\begin{array}{cc}
0 & A^{\prime}\\
A & 0
\end{array}\right)\Phi-F=0,
\]
which decomposes into two uncoupled equations. This, roughly speaking,
is the core rationale of variational calculus, which rigorously speaking
presupposes the normal solvability of equations of the form $Au=f$.
\end{enumerate}
\end{rem}

\begin{prop}
\label{prop:AeMaxPauli}We have 
\[
\mathcal{A}_{\mathrm{eMax}}\rightleftharpoons\left(\begin{array}{cc}
0_{2\times2} & \sigma_{1}\partial_{1}+\sigma_{2}\partial_{2}+\sigma_{3}\partial_{3}\\
-(\sigma_{1}\partial_{1}+\sigma_{2}\partial_{2}+\sigma_{3}\partial_{3}) & 0_{2\times2}
\end{array}\right),
\]
where 
\[
\sigma_{1}\coloneqq\left(\begin{array}{cc}
0 & 1\\
1 & 0
\end{array}\right),\quad\sigma_{2}\coloneqq\left(\begin{array}{cc}
0 & -\ii\\
\ii & 0
\end{array}\right),\quad\sigma_{3}\coloneqq\left(\begin{array}{cc}
1 & 0\\
0 & -1
\end{array}\right)
\]
denote the thus defined Pauli-matrices.
\end{prop}

\begin{proof}
Recall the notation 
\[
Q\left(0,\nabla\right)=\left(\begin{array}{cc}
0 & {\nabla}^{\top}\\
-\nabla & {\nabla}\times
\end{array}\right).
\]
Rearranging the terms according to the block structure of $\Phi_{\C}$,
we obtain
\begin{align*}
Q(0,\nabla) & =\left(\begin{array}{cc}
0 & \left(\begin{array}{ccc}
\partial_{1} & \partial_{2} & \partial_{3}\end{array}\right)\\
-\left(\begin{array}{c}
\partial_{1}\\
\partial_{2}\\
\partial_{3}
\end{array}\right) & \left(\begin{array}{ccc}
0 & -\partial_{3} & \partial_{2}\\
\partial_{3} & 0 & -\partial_{1}\\
-\partial_{2} & \partial_{1} & 0
\end{array}\right)
\end{array}\right)\\
 & =\left(\begin{array}{cc}
\left(\begin{array}{cc}
0 & \partial_{1}\\
-\partial_{1} & 0
\end{array}\right) & \left(\begin{array}{cc}
\partial_{2} & \partial_{3}\\
-\partial_{3} & \partial_{2}
\end{array}\right)\\
\left(\begin{array}{cc}
-\partial_{2} & \partial_{3}\\
-\partial_{3} & -\partial_{2}
\end{array}\right) & \left(\begin{array}{cc}
0 & -\partial_{1}\\
\partial_{1} & 0
\end{array}\right)
\end{array}\right).
\end{align*}
Hence, for a real-valued vector field $E=(E_{j})_{j\in\{0,1,2,3\}}$
we compute
\begin{align*}
\Phi_{\mathbb{C}}Q(0,\nabla)\Phi^{*}_{\mathbb{C}}\Phi_{\mathbb{C}}E & =\Phi_{\mathbb{C}}Q(0,\nabla)E\\
 & =\Phi_{\mathbb{C}}\text{\ensuremath{\left(\begin{array}{c}
\partial_{1}E_{1}+\partial_{2}E_{2}+\partial_{3}E_{3}\\
-\partial_{1}E_{0}-\partial_{3}E_{2}+\partial_{2}E_{3}\\
-\partial_{2}E_{0}+\partial_{3}E_{1}-\partial_{1}E_{3}\\
-\partial_{3}E_{0}-\partial_{2}E_{1}+\partial_{1}E_{2}
\end{array}\right)}}\\
 & =\partial_{1}\left(\begin{array}{c}
-\mathrm{i}E_{0}+E_{1}\\
-E_{3}+\mathrm{i}E_{2}
\end{array}\right)+\partial_{2}\left(\begin{array}{c}
\mathrm{i}E_{3}+E_{2}\\
-E_{0}-\mathrm{i}E_{1}
\end{array}\right)+\partial_{3}\left(\begin{array}{c}
-\mathrm{i}E_{2}+E_{3}\\
E_{1}-\mathrm{i}E_{0}
\end{array}\right)\\
 & =-\mathrm{i}\partial_{1}\sigma_{3}\Phi_{\mathbb{C}}E-\mathrm{i}\partial_{2}\sigma_{2}\Phi_{\mathbb{C}}E-\mathrm{i}\partial_{3}\sigma_{1}\Phi_{\mathbb{C}}E.
\end{align*}
Thus 
\[
\Phi_{\C}Q(0,\nabla)\Phi^{\ast}_{\C}=\left(\begin{array}{cc}
-\i\partial_{1} & \partial_{2}-\i\partial_{3}\\
-\partial_{2}-\i\partial_{3} & \i\partial_{1}
\end{array}\right)=-\i\left(\sigma_{3}\partial_{1}+\sigma_{2}\partial_{2}+\sigma_{1}\partial_{3}\right).
\]
Hence,
\[
\mathcal{A}_{\mathrm{eMax}}\rightleftharpoons\left(\begin{array}{cc}
0_{2\times2} & \i\left(\sigma_{3}\partial_{1}+\sigma_{2}\partial_{2}+\sigma_{1}\partial_{3}\right)\\
\i(\sigma_{3}\partial_{1}+\sigma_{2}\partial_{2}+\sigma_{1}\partial_{3}) & 0_{2\times2}
\end{array}\right).
\]
Applying now the unitary operator $\left(\begin{array}{cc}
1_{\mathbb{C}^{2}} & 0_{2\times2}\\
0_{2\times2} & \ii\,1_{\mathbb{C}^{2}}
\end{array}\right)$ we infer 
\begin{align*}
\mathcal{A}_{\mathrm{eMax}} & \rightleftharpoons\left(\begin{array}{cc}
1_{\mathbb{C}^{2}} & 0_{2\times2}\\
0_{2\times2} & \ii\,1_{\mathbb{C}^{2}}
\end{array}\right)\left(\begin{array}{cc}
0_{2\times2} & \i\left(\sigma_{3}\partial_{1}+\sigma_{2}\partial_{2}+\sigma_{1}\partial_{3}\right)\\
\i(\sigma_{3}\partial_{1}+\sigma_{2}\partial_{2}+\sigma_{1}\partial_{3}) & 0_{2\times2}
\end{array}\right)\left(\begin{array}{cc}
1_{\mathbb{C}^{2}} & 0_{2\times2}\\
0_{2\times2} & \ii\,1_{\mathbb{C}^{2}}
\end{array}\right)^{\ast}\\
 & =\left(\begin{array}{cc}
0_{2\times2} & \sigma_{3}\partial_{1}+\sigma_{2}\partial_{2}+\sigma_{1}\partial_{3}\\
-(\sigma_{3}\partial_{1}+\sigma_{2}\partial_{2}+\sigma_{1}\partial_{3}) & 0_{2\times2}
\end{array}\right)
\end{align*}
and by a coordinate change $\left(x_{1},x_{2},x_{3}\right)\mapsto\left(x_{3},x_{2},x_{1}\right)$
we finally arrive at
\[
\mathcal{A}_{\mathrm{eMax}}\rightleftharpoons\left(\begin{array}{cc}
0_{2\times2} & \sigma_{1}\partial_{1}+\sigma_{2}\partial_{2}+\sigma_{3}\partial_{3}\\
-(\sigma_{1}\partial_{1}+\sigma_{2}\partial_{2}+\sigma_{3}\partial_{3}) & 0_{2\times2}
\end{array}\right).\qedhere
\]
\end{proof}

In the established sense of unitary congruence, the extended Maxwell
equation \eqref{eq:Max-m=00003D0} therefore describes indeed the
same physics as the standard Dirac operator for massless particles.
In the light of this observation we feel that the real version \eqref{eq:Max-m=00003D0}
of the Dirac equation for massless particles might be an easier and
possibly more transparent approach to the fundamental equations of
quantum physics. After all there are more unitary congruences in the
real than in the complex case.

\section{\label{sec:Hilbert-Algebras-of}Hilbert Algebras of Real Matrices}

We note that the Pauli-matrices as they have occurred in Proposition
\ref{prop:AeMaxPauli} may be viewed as certain basis vectors and,
indeed, they are somewhat remindful of quaternions. To investigate
this deeper we need a bit more algebra, which, however, will turn
out to be much less complicated than the classical $\alpha,\beta,\gamma$
matrix considerations. Indeed, we shall in some respects not rely
on any particular choice of basis in our considerations, which will
turn out to simplify matters substantially. 

We start with the standard Hilbert algebra $\mathbb{R}^{\left(n+1\right)\times\left(n+1\right)}$
of $(n+1)\times(n+1)$-matrices, $n\in\mathbb{N}.$ The standard inner
product is given by 
\[
\left(A,B\right)\mapsto\frac{1}{n+1}\,\trace\left(A^{\top}B\right)
\]
and shows isometry to
\[
\mathbb{R}^{\left(n+1\right)^{2}}.
\]
We are particularly interested in $n=3$ and in a particular subalgebra
of $\R^{4\times4}$. 
\begin{defn*}
Set
\[
\mathbb{H}\coloneqq\left\{ Q\left(e_{0},e\right)|e_{0}\in\mathbb{R},e\in\mathbb{R}^{3}\right\} ,
\]
which are indeed the so-called \textbf{quaternions}. 

In the next statement, we will study matrix multiplication of quaternions.
It will turn out that $\mathbb{H}$ is a division ring (also called
skew-field).
\end{defn*}
\begin{lem}
\label{lem:Hcompurules}The set $\mathbb{H}$ is a subalgebra of $\R^{4\times4}$
with 
\[
Q(e_{0},e)Q(f_{0},f)=Q(e_{0}f_{0}-e^{\top}f,e_{0}f+f_{0}e+e\times f)
\]
and 
\[
Q(e_{0},e)^{-1}=\frac{1}{|(e_{0},e)|^{2}}Q(e_{0},-e)
\]
for $(e_{0},e)\ne0$, where $|\cdot|$ denotes the euclidean norm.
In particular, if $e,f\in\mathbb{R}^{3}$ with $e\bot f$, we have
\[
Q(0,e)Q(0,e)=-|e|^{2}1_{\mathbb{R}^{2}}\text{ and }Q(0,e)Q(0,f)=Q(0,e\times f)=-Q(0,f)Q(0,e).
\]
\end{lem}

\begin{proof}
For $e_{0},f_{0}\in\R$ and $e,f\in\R^{3}$ we compute 
\begin{align*}
 & Q(e_{0},e)Q(f_{0},f)\\
 & =\left(e_{0}1_{4\times4}+\left(\begin{array}{cc}
0 & e^{\top}\\
-e & e\times
\end{array}\right)\right)\left(f_{0}1_{4\times4}+\left(\begin{array}{cc}
0 & f^{\top}\\
-f & f\times
\end{array}\right)\right)\\
 & =e_{0}f_{0}1_{4\times4}+\left(\begin{array}{cc}
0 & f_{0}e^{\top}\\
-f_{0}e & f_{0}(e\times)
\end{array}\right)+\left(\begin{array}{cc}
0 & e_{0}f^{\top}\\
-e_{0}f & e_{0}(f\times)
\end{array}\right)+\left(\begin{array}{cc}
-e^{\top}f & e^{\top}\left(f\times\right)\\
-e\times f & -ef^{\top}+(e\times)(f\times)
\end{array}\right).
\end{align*}
Now an easy computation reveals
\begin{align*}
(e\times)(f\times) & =\left(\begin{array}{ccc}
0 & -e_{3} & e_{2}\\
e_{3} & 0 & -e_{1}\\
-e_{2} & e_{1} & 0
\end{array}\right)\left(\begin{array}{ccc}
0 & -f_{3} & f_{2}\\
f_{3} & 0 & -f_{1}\\
-f_{2} & f_{1} & 0
\end{array}\right)\\
 & =-e^{\top}f1_{3\times3}+fe^{\top}
\end{align*}
and thus,
\begin{align*}
 & Q(e_{0},e)Q(f_{0},f)\\
 & =\left(e_{0}f_{0}-e^{\top}f\right)1_{4\times4}+\left(\begin{array}{cc}
0 & e_{0}f^{\top}+f_{0}e^{\top}+\left(e\times f\right)^{\top}\\
-(e_{0}f+f_{0}e+e\times f) & e_{0}(f\times)+f_{0}(e\times)+fe^{\top}-ef^{\top}
\end{array}\right)\\
 & =Q(e_{0}f_{0}-e^{\top}f,e_{0}f+f_{0}e+e\times f),
\end{align*}
where we have used $e^{\top}(f\times)=(e\times f)^{\top}$ and 
\[
\left(e\times f\right)\times=fe^{\top}-ef^{\top}.
\]
Using this formula, we get that 
\[
Q(e_{0},e)Q(e_{0},-e)=Q(e^{2}_{0}+e^{\top}e,0)=|(e_{0},e)|^{2}1_{4\times4}
\]
showing the formula for the inverse. The remaining formulas are direct
consequences of the first. 
\end{proof}

\begin{rem}
The space $\mathbb{H}$ can be constructed from $\mathbb{R}$ by two
steps of the Cayley--Dickson process of dimension doubling of {*}-algebras\footnote{The matrix formulation allows to conveniently use the usual matrix
product as the operation in the new $*$-algebra for the first two
steps starting from $\mathbb{R}$.}
\[
\left(A,B\right)\mapsto\left(\begin{array}{cc}
A & -B^{*}\\
B & A^{*}
\end{array}\right).
\]
The first step
\begin{align*}
\mathbb{R}^{2} & \to\mathbb{C}\\
\left(\begin{array}{c}
x\\
y
\end{array}\right) & \mapsto\left(\begin{array}{cc}
x & -y\\
y & x
\end{array}\right)
\end{align*}
produces the complex numbers from pairs of reals. The second step
\begin{align*}
\mathbb{C}^{2} & \to\mathbb{\mathbb{H}}\\
\left(\begin{array}{c}
z\\
\zeta
\end{array}\right) & \mapsto\left(\begin{array}{cc}
z & -\zeta^{*}\\
\zeta & z^{*}
\end{array}\right)
\end{align*}
leads to quaternions from pairs of complex ``numbers''. Here $z^{\ast}=x-\mathrm{i}y$
is just complex conjugation of $z=x+\mathrm{i}y$. With $z=\left(\begin{array}{cc}
x & -y\\
y & x
\end{array}\right)=x+y\,\ii$ and $\zeta=\left(\begin{array}{cc}
\xi & -\psi\\
\psi & \xi
\end{array}\right)=\xi+\psi\,\ii$ we have
\begin{align*}
\left(\begin{array}{cc}
z & -\zeta^{*}\\
\zeta & z^{*}
\end{array}\right) & =\left(\begin{array}{cc}
\left(\begin{array}{cc}
x & -y\\
y & x
\end{array}\right) & \left(\begin{array}{cc}
-\xi & -\psi\\
\psi & -\xi
\end{array}\right)\\
\left(\begin{array}{cc}
\xi & -\psi\\
\psi & \xi
\end{array}\right) & \left(\begin{array}{cc}
x & y\\
-y & x
\end{array}\right)
\end{array}\right)\\
 & =\left(\begin{array}{cccc}
0 & -y & -\xi & -\psi\\
y & 0 & \psi & -\xi\\
\xi & -\psi & 0 & y\\
\psi & \xi & -y & 0
\end{array}\right)+\left(\begin{array}{cccc}
x & 0 & 0 & 0\\
0 & x & 0 & 0\\
0 & 0 & x & 0\\
0 & 0 & 0 & x
\end{array}\right)\\
 & =Q\left(x,-\left(\begin{array}{c}
y\\
\xi\\
\psi
\end{array}\right)\right).
\end{align*}
\end{rem}

Of course, anything unitarily congruent would produce another representation
of the quaternions. The example of interest is here the Lorentz signature
matrix
\[
\sigma\coloneqq\left(\begin{array}{cc}
1 & \left(\begin{array}{ccc}
0 & 0 & 0\end{array}\right)\\
\left(\begin{array}{c}
0\\
0\\
0
\end{array}\right) & \left(\begin{array}{ccc}
-1 & 0 & 0\\
0 & -1 & 0\\
0 & 0 & -1
\end{array}\right)
\end{array}\right)
\]
as a unitary congruence, which is unitary and selfadjoint. 
\begin{defn*}
For $e_{0}\in\R$ and $e\in\R^{3}$ we define
\begin{align*}
Q^{\prime}\left(e_{0},e\right) & \coloneqq\sigma Q\left(e_{0},e\right)\sigma\\
 & =\left(\begin{array}{cc}
e_{0} & -e^{\top}\\
e & e_{0}1_{3\times3}+e\times
\end{array}\right)=e_{0}1_{4\times4}+Q^{\prime}\left(0,e\right)
\end{align*}
and set 
\[
\mathbb{H}'\coloneqq\{Q'(e_{0},e)\,;\,e_{0}\in\R,e\in\R^{3}\}
\]
as the set of \textbf{co-quaternions}. 
\end{defn*}
\begin{rem}
By unitary congruence it is clear that $\mathbb{H}'$ is also a division
ring. 
\end{rem}

As we shall see both $\mathbb{H}^{\prime}$ and $\mathbb{H}$ will
be of importance here. 
\begin{prop}
\label{prop:commute_Q_Q'}Let $e_{0},f_{0}\in\R$ and $e,f\in\R^{3}$.
Then 
\[
Q^{\prime}\left(e_{0},e\right)Q\left(f_{0},f\right)=Q\left(f_{0},f\right)Q^{\prime}\left(e_{0},e\right)
\]
and hence, in particular the matrix 
\[
Q^{\prime}\left(e_{0},e\right)Q\left(f_{0},f\right)
\]
is selfadjoint. Furthermore, if $|f|=|e|=1,$ we have
\[
Q(0,f)^{2}=-1_{4\times4}=Q'(0,e)^{2}.
\]
Moreover, if $e\bot f$ then 
\[
Q^{\prime}\left(0,e\right)Q^{\prime}\left(0,f\right)\text{ and }Q\left(0,e\right)Q\left(0,f\right)
\]
are skew-selfadjoint. 
\end{prop}

\begin{proof}
As $Q'(e_{0},e)=e_{0}1_{4\times4}+Q'(0,e)$ and $Q(f_{0},f)=f_{0}1_{4\times4}+Q(0,f)$
and both $Q'(0,e)$ and $Q(0,f)$ are skew-selfadjoint, the selfadjointness
of $Q^{\prime}\left(e_{0},e\right)Q\left(f_{0},f\right)$ follows
from the commutator relation by evaluating $(Q^{\prime}\left(e_{0},e\right)Q\left(f_{0},f\right))^{*}$
and $(Q\left(f_{0},f\right)Q^{\prime}\left(e_{0},e\right))^{*}$.
Moreover, it suffices to show $Q'(0,e)Q(0,f)=Q(0,f)Q'(0,e).$ We compute
\begin{align*}
Q'(0,e)Q(0,f) & =\left(\begin{array}{cc}
0 & -e^{\top}\\
e & e\times
\end{array}\right)\left(\begin{array}{cc}
0 & f^{\top}\\
-f & f\times
\end{array}\right)\\
 & =\left(\begin{array}{cc}
e^{\top}f & -e^{\top}\left(f\times\right)\\
-e\times f & ef^{\top}+(e\times)(f\times)
\end{array}\right)\\
 & =\left(\begin{array}{cc}
e^{\top}f & \left(f\times e\right)^{\top}\\
f\times e & ef^{\top}-e^{\top}f1_{3\times3}+fe^{\top}
\end{array}\right)\\
 & =\left(\begin{array}{cc}
f^{\top}e & f^{\top}e\times\\
f\times e & fe^{\top}+(f\times)(e\times)
\end{array}\right)\\
 & =\left(\begin{array}{cc}
0 & f^{\top}\\
-f & f\times
\end{array}\right)\left(\begin{array}{cc}
0 & -e^{\top}\\
e & e\times
\end{array}\right)=Q(0,f)Q'(0,e)
\end{align*}
showing the claim. %
{} Moreover, by Lemma \ref{lem:Hcompurules}
\[
Q(0,f)^{2}=Q(-|f|^{2},0)=-|f|^{2}1_{4\times4}
\]
and hence also, 
\[
Q'(0,e)^{2}=\sigma Q(0,e)^{2}\sigma=-|e|^{2}1_{4\times4}.
\]
Assume now $e\bot f$. Then, by Lemma \ref{lem:Hcompurules},
\[
Q(0,e)Q(0,f)=Q(0,e\times f)=-Q(0,e\times f)^{*}
\]
establishing skew-selfadjointness. As $Q^{\prime}\left(0,e\right)Q^{\prime}\left(0,f\right)=\sigma Q(0,e)Q(0,f)\sigma$
is unitarily equivalent to $Q(0,e)Q(0,f)$ the former is skew-selfadjoint
since the latter is.
\end{proof}

\begin{rem}[Canonical basis for quaternions]
If we assume that $e,f\in\R^{3}$ are orthonormal, we get that also
$e\times f$ is normalised and $\left(e,f,e\times f\right)$ is a
right-handed orthonormal system. Moreover, since
\[
Q\left(0,e\right)^{2}=-1=Q^{\prime}\left(0,e\right)^{2},
\]
all such elements can serve as a complex unit. Indeed, choosing in
particular the canonical basis $\left(e_{1},e_{2},e_{3}\right)$
\[
e_{1}=\left(\begin{array}{c}
1\\
0\\
0
\end{array}\right),\:e_{2}=\left(\begin{array}{c}
0\\
1\\
0
\end{array}\right),\:e_{3}=\left(\begin{array}{c}
0\\
0\\
1
\end{array}\right)=e_{1}\times e_{2},
\]
 of $\mathbb{R}^{3}$, we recover 
\begin{align*}
\mathfrak{i} & \coloneqq-Q\left(0,e_{1}\right)=\left(\begin{array}{cccc}
0 & -1 & 0 & 0\\
1 & 0 & 0 & 0\\
0 & 0 & 0 & 1\\
0 & 0 & -1 & 0
\end{array}\right)=\left(\begin{array}{cc}
\ii & 0\\
0 & -\ii
\end{array}\right),\\
\mathfrak{j} & \coloneqq Q\left(0,e_{2}\right)=\left(\begin{array}{cccc}
0 & 0 & 1 & 0\\
0 & 0 & 0 & 1\\
-1 & 0 & 0 & 0\\
0 & -1 & 0 & 0
\end{array}\right)=\left(\begin{array}{cc}
0 & 1\\
-1 & 0
\end{array}\right),\\
\mathfrak{k} & \coloneqq-Q\left(0,e_{3}\right)=\left(\begin{array}{cccc}
0 & 0 & 0 & -1\\
0 & 0 & 1 & 0\\
0 & -1 & 0 & 0\\
1 & 0 & 0 & 0
\end{array}\right)=\left(\begin{array}{cc}
0 & \ii\\
\ii & 0
\end{array}\right),
\end{align*}
the commonly used orthonormal basis $\left(\mathfrak{i},\mathfrak{j},\mathfrak{k}\right)$
of the pure quaternions, that is, for the ortho-complement $\mathbb{R}^{\perp_{\mathbb{H}}}$
of $\mathbb{R}$ in $\mathbb{H}$. The analogous result follows, of
course, with $Q$ replaced by $Q^{\prime}$. %
\end{rem}

\begin{rem}[Exponentials of quaternions]
We recall that we have a Hilbert space of normal matrices 
\[
\left\{ Q\left(e_{0},e\right)|\left(e_{0},e\right)\in\mathbb{R}\times\mathbb{R}^{3}=\mathbb{R}^{4}\right\} 
\]
(recall from Lemma \ref{lem:Hcompurules}, $Q\left(e_{0},e\right)^{*}=Q\left(e_{0},-e\right)=|(e_{0},e)|^{2}Q\left(e_{0},e\right)^{-1}$
if $(e_{0},e)\neq0$) which has the commutator as a Lie product (Lie
bracket). A derived construct which appears to be of interest is another
Lie subgroup of matrices
\[
\left\{ \exp\left(\alpha\,Q\left(e_{0},e\right)\right)|\left(e_{0},e\right)\in\mathbb{R}\times\mathbb{R}^{3}=\mathbb{R}^{4},\alpha\in\mathbb{R}\right\} .
\]
Since
\begin{align*}
Q\left(e_{0},e\right)^{2} & =Q\left(e^{2}_{0}-e^{\top}e,2e_{0}e\right)\\
Q\left(0,e\right)^{2} & =Q\left(-e^{\top}e,0\right)\\
 & =-e^{\top}e1_{4\times4}=-|e|^{2}1_{4\times4},
\end{align*}
we get, with $\operatorname{sinc}(x)=\sin(x)/x$ and the power series
expansions of $\cos$ and $\sin$, 
\begin{align*}
\exp\left(\alpha Q\left(e_{0},e\right)\right) & =\exp\left(\alpha e_{0}\right)\exp\left(\alpha Q\left(0,e\right)\right)\\
 & =\exp\left(\alpha e_{0}\right)\sum^{\infty}_{k=0}\frac{\alpha^{k}}{k!}Q(0,e)^{k}\\
 & =\exp\left(\alpha e_{0}\right)\left(\sum^{\infty}_{k=0}\frac{(-1)^{k}(\alpha|e|)^{2k}}{(2k)!}1_{4\times4}+\sum^{\infty}_{k=0}\frac{(-1)^{k}\alpha^{2k+1}|e|^{2k}}{(2k+1)!}Q(0,e)\right)\\
 & =\exp(\alpha e_{0})\left(\cos(\alpha|e|)1_{4\times4}+\frac{1}{|e|}\sin(\alpha|e|)Q(0,e)\right)\\
 & =\exp\left(\alpha e_{0}\right)Q\left(\cos\left(\alpha|e|\right),\sinc\left(\alpha|e|\right)\alpha e\right).
\end{align*}
So, not surprisingly, the result is just another quaternion. In particular,
for $|e|=1$ we obtain
\[
\exp(\alpha Q(0,e))=Q(\cos(\alpha),\sin(\alpha)e)
\]
and hence, 
\[
\left\{ \exp\left(\alpha\,Q\left(0,e\right)\right)|e\in\mathbb{R}^{3},|e|=1,\,\alpha\in\left[0,\pi\right]\right\} =\left\{ Q\left(E\right)|E\in\mathbb{R}^{4},|E|=1\right\} .
\]
Indeed, that $\left\{ \exp\left(\alpha\,Q\left(0,e\right)\right)|e\in\mathbb{R}^{3},\left|e\right|=1,\,\alpha\in\left[0,\pi\right]\right\} \subseteq\left\{ Q\left(E\right)|E\in\mathbb{R}^{4},\left|E\right|=1\right\} $
is obvious. Conversely, if $F=\left(\begin{array}{c}
f_{0}\\
f
\end{array}\right)$ with $|F|=1$ then $f_{0}\in[-1,1]$. So, there is an $\alpha\in[0,\pi]$
such that
\[
f_{0}=\cos\left(\alpha\right).
\]
then 
\[
|f|^{2}=|F|^{2}-f^{2}_{0}=\left(\sin\left(\alpha\right)\right)^{2}.
\]
For $\alpha\not\in\left\{ 0,\pi\right\} $ we see that
\begin{equation}
e=\frac{1}{\sin\left(\alpha\right)}f\label{eq:esinalpha}
\end{equation}
has norm $1$. Then 
\[
Q\left(F\right)=Q\left(\cos\left(\alpha\right),\sin\left(\alpha\right)e\right)=\exp\left(\alpha\,Q\left(0,e\right)\right).
\]
If $\alpha\in\{0,\pi\}$ we have $f=0$ and $f_{0}=\pm1$. In this
case we have either
\[
Q\left(F\right)=1=\exp\left(0\,Q\left(0,e\right)\right)\text{ or }Q\left(F\right)=-1=\exp\left(\pi\,Q\left(0,e\right)\right),
\]
for any normalised $e$.

Since for $E=(e_{0},e)\in\mathbb{R}^{4}$
\begin{align}
Q\left(E\right)^{*}Q\left(E\right) & =\left(e_{0}-Q\left(0,e\right)\right)\left(e_{0}+Q\left(0,e\right)\right)\nonumber \\
 & =e^{2}_{0}-Q\left(0,e\right)^{2}\label{eq:unita}\\
 & =e^{2}_{0}+|e|^{2}\nonumber \\
 & =|E|^{2},\nonumber 
\end{align}
we see that
\[
\left\{ Q\left(E\right)|E\in\mathbb{R}^{4},|E|=1\right\} 
\]
is a set of (real) unitary matrices. Moreover, in general,
\begin{align*}
\left(Q\left(E\right)Q\left(F\right)\right)^{\top}Q\left(E\right)Q\left(F\right) & =Q\left(F\right)^{\top}Q\left(E\right)^{\top}Q\left(E\right)Q\left(F\right)\\
 & =|E|^{2}|F|^{2}
\end{align*}
 and hence, if
\[
Q\left(E\right)Q\left(F\right)=Q\left(Z\right)
\]
with $|E|=|F|=1$, then $|Z|=1$, and hence, there is a $z\in\mathbb{R}^{3}$,
$|z|=1$, and an $\alpha\in\left[0,\pi\right]$ such that $Q\left(Z\right)=Q\left(\cos\left(\alpha\right),\sin\left(\alpha\right)z\right)=\exp\left(\alpha Q\left(0,z\right)\right)$.
This gives 
\[
\exp\left(\beta\,Q\left(0,e\right)\right)\exp\left(\gamma\,Q\left(0,f\right)\right)=\exp\left(Q\left(0,\beta e\right)\right)\exp\left(Q\left(0,\gamma f\right)\right)=\exp\left(\alpha\,Q\left(0,z\right)\right)=\exp\left(Q\left(0,\alpha z\right)\right).
\]
Note that we have an explicit construction of $\alpha$ and $z$ (see
\eqref{eq:esinalpha}), in effect conveniently bypassing the usual
considerations via series of higher commutators by the Baker--Campell--Hausdorff
formula. From \eqref{eq:unita} we have in particular that $\left\{ Q\left(E\right)|E\in\mathbb{R}^{4},|E|=1\right\} $
is a subgroup of unitary operators in the Hilbert space $\mathbb{H}$.
From this observation one realises that exponentials of pure quaternions
are not really more particular as exponentials of general quaternions.
\end{rem}

\section{The Dirac Equation including the mass term}

As we have seen the extended Maxwell operator is indeed unitarily
congruent to the classical Dirac operator for massless particles.
The operator has the form
\[
\mathcal{A}_{\mathrm{eMax}}=\left(\begin{array}{cc}
0 & -Q\left(0,\nabla\right)\\
Q\left(0,\nabla\right) & 0
\end{array}\right),
\]
thus only making use of the standard quaternion representation $\mathbb{H}$
introduced above.

Remarkably this reconciles Maxwell's quaternionic approach to electrodynamics
with Heaviside's and Gibbs' perspective, which reshaped Maxwell's
quaternion setting in such an unfortunate way to eventually force
Dirac to reinvent biquaternions under different names to arrive at
his factorisation of the Klein--Gordon equation.

Recalling $e_{1}=(1,0,0)$, let us inspect
\[
Q^{\prime}\left(0,e_{1}\right)=\left(\begin{array}{cc}
0 & -e^{\top}_{1}\\
e_{1} & e_{1}\times
\end{array}\right)=\left(\begin{array}{cc}
\left(\begin{array}{cc}
0 & -1\\
1 & 0
\end{array}\right) & \left(\begin{array}{cc}
0 & 0\\
0 & 0
\end{array}\right)\\
\left(\begin{array}{cc}
0 & 0\\
0 & 0
\end{array}\right) & \left(\begin{array}{cc}
0 & -1\\
1 & 0
\end{array}\right)
\end{array}\right)=\left(\begin{array}{cc}
\ii & 0\\
0 & \ii
\end{array}\right).
\]
Apparently, $Q^{\prime}\left(0,e_{1}\right)$ is simply the multiplication
by the imaginary unit $\ii$, which we introduced in the above to
bring the operator even closer to the more common formulation of the
Dirac operator. We have indeed
\begin{align*}
 & \left(\begin{array}{cc}
1 & 0\\
0 & Q^{\prime}\left(0,e_{1}\right)
\end{array}\right)\left(\begin{array}{cc}
0 & -Q\left(0,\nabla\right)\\
Q\left(0,\nabla\right) & 0
\end{array}\right)\left(\begin{array}{cc}
1 & 0\\
0 & Q^{\prime}\left(0,e_{1}\right)
\end{array}\right)^{*}\\
 & =\left(\begin{array}{cc}
0 & Q\left(0,\nabla\right)Q^{\prime}\left(0,e_{1}\right)\\
Q^{\prime}\left(0,e_{1}\right)Q\left(0,\nabla\right) & 0
\end{array}\right)\\
 & =\left(\begin{array}{cc}
0 & Q^{\prime}\left(0,e_{1}\right)Q\left(0,\nabla\right)\\
Q^{\prime}\left(0,e_{1}\right)Q\left(0,\nabla\right) & 0
\end{array}\right),
\end{align*}
where we have used \prettyref{prop:commute_Q_Q'} and hence,
\[
\mathcal{A}_{\mathrm{eMax}}\rightleftharpoons\left(\begin{array}{cc}
0 & Q^{\prime}\left(0,e_{1}\right)Q\left(0,\nabla\right)\\
Q^{\prime}\left(0,e_{1}\right)Q\left(0,\nabla\right) & 0
\end{array}\right)
\]
where the operator on the right-hand side provides the classical formulation
of Dirac's operator for the case of vanishing mass in purely real
terms\footnote{It seems that the use of complex numbers occasionally tends to obscure
the actually underlying structure. An example may be the Schrödinger
operator $\partial_{0}+\ii\,\Delta$, which is in real terms just
a first order version of a 3-dimensional Kirchhoff plate operator
$\partial^{2}_{0}+\Delta^{2}=\left(\partial_{0}+\ii\Delta\right)\left(\partial_{0}-\ii\Delta\right).$
Indeed, separating real and imaginary part in the Schrödinger equation
yields the system
\[
\left(\partial_{0}+\left(\begin{array}{cc}
0 & -\Delta\\
\Delta & 0
\end{array}\right)\right)\left(\begin{array}{c}
u\\
v
\end{array}\right)=\left(\begin{array}{c}
f\\
g
\end{array}\right).
\]
Eliminating $v$ from this, we get
\begin{align*}
\partial^{2}_{0}u-\Delta\left(-\Delta u+g\right) & =\partial_{0}f\\
\left(\partial^{2}_{0}+\Delta^{2}\right)u & =\partial_{0}f+\Delta g,
\end{align*}
which is indeed a 3-dimensional version of the Kirchhoff plate equation.}. 
\begin{lem}
\label{lem:orthogonal}Let $R\in\mathrm{SO}(3)$ ; i.e., $RR^{\top}=R^{\top}R=1_{3\times3}$
and $\det R=1$. Then 
\[
\left(\begin{array}{cc}
1 & 0\\
0 & R
\end{array}\right)Q^{\prime}\left(0,f\right)\left(\begin{array}{cc}
1 & 0\\
0 & R
\end{array}\right)^{\top}=Q^{\prime}\left(0,Rf\right),\left(\begin{array}{cc}
1 & 0\\
0 & R
\end{array}\right)Q\left(0,f\right)\left(\begin{array}{cc}
1 & 0\\
0 & R
\end{array}\right)^{\top}=Q\left(0,Rf\right)
\]
for each $f\in\R^{3}$.
\end{lem}

\begin{proof}
For the first equality we compute
\begin{align*}
\left(\begin{array}{cc}
1 & 0\\
0 & R
\end{array}\right)Q^{\prime}\left(0,f\right)\left(\begin{array}{cc}
1 & 0\\
0 & R^{\top}
\end{array}\right) & =\left(\begin{array}{cc}
1 & 0\\
0 & R
\end{array}\right)\left(\begin{array}{cc}
0 & -f^{\top}\\
f & f\times
\end{array}\right)\left(\begin{array}{cc}
1 & 0\\
0 & R^{\top}
\end{array}\right)\\
 & =\left(\begin{array}{cc}
0 & -f^{\top}R^{\top}\\
Rf & R(f\times)R^{\top}
\end{array}\right)\\
 & =\left(\begin{array}{cc}
0 & -\left(Rf\right)^{\top}\\
Rf & R(f\times)R^{\top}
\end{array}\right).
\end{align*}
Now for $g\in\R^{3}$ we compute 
\[
R(f\times g)=\left(\begin{array}{c}
\langle R^{\top}e_{1},f\times g\rangle\\
\langle R^{\top}e_{2},f\times g\rangle\\
\langle R^{\top}e_{3},f\times g\rangle
\end{array}\right)=\left(\begin{array}{c}
\det(R^{\top}e_{1}\,f\,g)\\
\det(R^{\top}e_{2}\,f\,g)\\
\det(R^{\top}e_{3}\,f\,g)
\end{array}\right)=\left(\det R^{\top}\right)\left(\begin{array}{c}
\det(e_{1}\,Rf\,Rg)\\
\det(e_{2}\,Rf\,Rg)\\
\det(e_{3}\,Rf\,Rg)
\end{array}\right)=(Rf\times Rg)
\]
and hence, $R(f\times)R^{\top}=(Rf)\times.$ Hence, the first formula
follows. The rationale for $Q$ instead of $Q'$ follows the same
lines.
\end{proof}

\begin{cor}
\label{cor:orthogonal_derivative}Let $R\in\mathrm{SO}(3)$. Then
\[
\left(\begin{array}{cc}
1 & 0\\
0 & R
\end{array}\right)Q\left(0,\nabla\right)\left(\begin{array}{cc}
1 & 0\\
0 & R^{\top}
\end{array}\right)=Q\left(0,R\nabla\right).
\]
\end{cor}

\begin{proof}
The latter follows from \prettyref{lem:orthogonal} by applying the
Fourier transform $\mathcal{F}$. Indeed, for $u$ smooth 
\begin{align*}
\left(\mathcal{F}\left(\begin{array}{cc}
1 & 0\\
0 & R
\end{array}\right)Q\left(0,\nabla\right)\left(\begin{array}{cc}
1 & 0\\
0 & R^{\top}
\end{array}\right)u\right)(\xi) & =\left(\begin{array}{cc}
1 & 0\\
0 & R
\end{array}\right)\left(\mathcal{F}\left(\begin{array}{cc}
0 & {\nabla}^{\top}\\
-\nabla & {\nabla}\times
\end{array}\right)\left(\begin{array}{cc}
1 & 0\\
0 & R^{\top}
\end{array}\right)u\right)(\xi)\\
 & =\left(\begin{array}{cc}
1 & 0\\
0 & R
\end{array}\right)\left(\begin{array}{cc}
0 & {\i\xi}^{\top}\\
-\i\xi & {\i\xi}\times
\end{array}\right)\left(\begin{array}{cc}
1 & 0\\
0 & R^{\top}
\end{array}\right)\left(\mathcal{F}u\right)(\xi)\\
 & =\i\left(\begin{array}{cc}
1 & 0\\
0 & R
\end{array}\right)Q(0,\xi)\left(\begin{array}{cc}
1 & 0\\
0 & R^{\top}
\end{array}\right)\left(\mathcal{F}u\right)(\xi)\\
 & =\i Q(0,R\xi)\left(\mathcal{F}u\right)(\xi)\\
 & =Q(0,R\i\xi)\left(\mathcal{F}u\right)(\xi)=\mathcal{F}(Q(0,R\nabla)u)(\xi)
\end{align*}
for each $\xi\in\R^{3}$, which yields the claim. 
\end{proof}

\begin{defn*}
For $A\in\mathbb{R}^{3\times3}$ invertible we define the transformation
for any $m\in\mathbb{N}_{>0}$
\begin{align*}
\chi_{A}:L^{2}\left(\mathbb{R}^{3},\mathbb{R}^{m}\right) & \to L^{2}\left(\mathbb{R}^{3},\mathbb{R}^{m}\right)\\
u & \mapsto\sqrt{\left|\det\left(A\right)\right|}\:u\circ A.
\end{align*}
\end{defn*}
\begin{rem}
Note that $\chi_{A}$ is unitary, since 
\[
\int_{\R^{3}}|\left(\chi_{A}u\right)(x)|^{2}\d x=\int_{\R^{3}}|\det A||u(Ax)|^{2}\d x=\int_{\R^{3}}|u(x)|^{2}\d x
\]
for $u\in L^{2}(\R^{3};\R^{m})$ and its inverse is clearly given
by $\chi_{A^{-1}}.$ 
\end{rem}

\begin{prop}
Let $e\in\R^{3}$ be normalised. Then 
\[
\mathcal{A}_{\mathrm{eMax}}\rightleftharpoons\left(\begin{array}{cc}
0 & Q^{\prime}\left(0,e\right)Q\left(0,\nabla\right)\\
Q^{\prime}\left(0,e\right)Q\left(0,\nabla\right) & 0
\end{array}\right).
\]
\end{prop}

\begin{proof}
We have already seen that this congruence holds for $e=e_{1}.$ Now
we choose a matrix $R\in\mathrm{SO}(3)$ such that $e=Re_{1}$. By
\prettyref{lem:orthogonal} and \prettyref{cor:orthogonal_derivative}
we infer 
\begin{align*}
 & \left(\begin{array}{cc}
0 & Q^{\prime}\left(0,e_{1}\right)Q\left(0,\nabla\right)\\
Q^{\prime}\left(0,e_{1}\right)Q\left(0,\nabla\right) & 0
\end{array}\right)\\
 & \rightleftharpoons\left(\begin{array}{cc}
\left(\begin{array}{cc}
1 & 0\\
0 & R
\end{array}\right) & 0\\
0 & \left(\begin{array}{cc}
1 & 0\\
0 & R
\end{array}\right)
\end{array}\right)\left(\begin{array}{cc}
0 & Q^{\prime}\left(0,e_{1}\right)Q\left(0,\nabla\right)\\
Q^{\prime}\left(0,e_{1}\right)Q\left(0,\nabla\right) & 0
\end{array}\right)\left(\begin{array}{cc}
\left(\begin{array}{cc}
1 & 0\\
0 & R^{\top}
\end{array}\right) & 0\\
0 & \left(\begin{array}{cc}
1 & 0\\
0 & R^{\top}
\end{array}\right)
\end{array}\right)\\
 & \rightleftharpoons\left(\begin{array}{cc}
0 & Q^{\prime}\left(0,e\right)Q\left(0,R\nabla\right)\\
Q^{\prime}\left(0,e\right)Q\left(0,R\nabla\right) & 0
\end{array}\right).
\end{align*}
Next we show 
\begin{equation}
\chi_{R^{\top}}Q(0,R\nabla)\chi_{R}=Q(0,\nabla)\label{eq:transform_Q}
\end{equation}
For doing so, we again apply the Fourier transform $\mathcal{F}$
for a smooth function $u$. First we observe that 
\[
\left(\mathcal{F}\chi_{R}u\right)(\xi)=\frac{1}{(2\pi)^{\frac{3}{2}}}\int_{\R^{3}}\e^{-\i\xi\cdot x}u(Rx)\d x=\frac{1}{(2\pi)^{\frac{3}{2}}}\int_{\R^{3}}\e^{-\i\xi\cdot R^{\top}x}u(x)\d x=\left(\mathcal{F}u\right)(R\xi)
\]
and hence, $\mathcal{F}\chi_{R}=\chi_{R}\mathcal{F}.$ This yields
\begin{align*}
\left(\mathcal{F}\chi_{R^{\top}}Q(0,R\nabla)\chi_{R}u\right)(\xi) & =\left(\mathcal{F}Q(0,R\nabla)\chi_{R}u\right)(R^{\top}\xi)\\
 & =Q(0,R\i R^{\top}\xi)\left(\mathcal{F}\chi_{R}u\right)(R^{\top}\xi)\\
 & =Q(0,\i\xi)\left(\mathcal{F}u\right)(\xi)\\
 & =\mathcal{F}(Q(0,\nabla)u)(\xi)
\end{align*}
and hence, \prettyref{eq:transform_Q} follows. Since clearly $\chi_{R^{\top}}Q'(0,e)\chi_{R}=Q'(0,e)$,
we infer 
\[
\left(\begin{array}{cc}
0 & Q^{\prime}\left(0,e\right)Q\left(0,R\nabla\right)\\
Q^{\prime}\left(0,e\right)Q\left(0,R\nabla\right) & 0
\end{array}\right)\rightleftharpoons\left(\begin{array}{cc}
0 & Q^{\prime}\left(0,e\right)Q\left(0,\nabla\right)\\
Q^{\prime}\left(0,e\right)Q\left(0,\nabla\right) & 0
\end{array}\right).\qedhere
\]
\end{proof}

All in all we see that via unitary congruence the role of $e_{1}$
in the Dirac operator can be replaced by any normalised vector $e$
and we shall do so in the following, which by unitary congruence in
fact lead to the same physics. In fact, we also have an invariance
property that appears to be of interest.
\begin{prop}
\label{prop:UeInv}The operator
\[
\left(\begin{array}{cc}
0 & Q^{\prime}\left(0,e\right)Q\left(0,\nabla\right)\\
Q^{\prime}\left(0,e\right)Q\left(0,\nabla\right) & 0
\end{array}\right)
\]
is\textbf{\textup{ $U_{e}(1)$-invariant}}, that is, for all $\alpha\in\mathbb{R}$
it commutes with 
\[
\exp\left(\alpha\left(\begin{array}{cc}
Q^{\prime}\left(0,e\right) & 0\\
0 & Q^{\prime}\left(0,e\right)
\end{array}\right)\right).
\]
\end{prop}

\begin{proof}
The proof is a consequence of the power series expansion of $\exp$
and the fact that the operator considered commutes with $\left(\begin{array}{cc}
Q^{\prime}\left(0,e\right) & 0\\
0 & Q^{\prime}\left(0,e\right)
\end{array}\right)$ (and, hence, with all powers of it).
\end{proof}

\begin{rem}
\label{rem:Ue1-inv}(a) From the proof we see that only commutativity
with $Q^{\prime}\left(0,e\right)$ is needed to obtain the statement.
In fact, take $e,f\in\mathbb{R}^{3}$, then by Lemma \ref{lem:Hcompurules},
\begin{align*}
Q^{\prime}\left(0,e\right)Q^{\prime}\left(0,f\right)=Q^{\prime}\left(0,f\right)Q^{\prime}\left(0,e\right) & \iff Q^{\prime}\left(-e^{\top}f,e\times f\right)=Q^{\prime}\left(-f^{\top}e,f\times e\right)\\
 & \iff e\times f=f\times e\\
 & \iff e\times f=0\\
 & \iff e\parallel f.
\end{align*}
Thus, for normalised $e,f\in\mathbb{R}^{3}$, we obtain 
\[
\left(\begin{array}{cc}
0 & Q^{\prime}\left(0,e\right)Q\left(0,\nabla\right)\\
Q^{\prime}\left(0,e\right)Q\left(0,\nabla\right) & 0
\end{array}\right)\text{ is }U_{f}(1)\text{-invariant}\ \iff\ f=\pm e.
\]

(b) In physics literature one can find the notion $U(1)$-invariance,
which coincides with $U_{e_{1}}(1)$-invariance in \prettyref{prop:UeInv}.
Note that by (a), demanding for $U(1)$-invariance is asking the Dirac
operator to be of special form. Since the above rationale demonstrates
that changing $e$ does not change the physics, $U(1)$-invariance
here is biased towards $e_{1}$ only for aesthetic rather than for
physics reasons. The reason for this choice may be caused by the following
lucky accident: The $8\times8$ matrix 
\begin{align*}
\left(\begin{array}{cc}
Q^{\prime}\left(0,e_{1}\right) & 0\\
0 & Q^{\prime}\left(0,e_{1}\right)
\end{array}\right) & =\left(\begin{array}{cc}
\left(\begin{array}{cccc}
0 & -1 & 0 & 0\\
1 & 0 & 0 & 0\\
0 & 0 & 0 & -1\\
0 & 0 & 1 & 0
\end{array}\right) & \left(\begin{array}{cccc}
0 & 0 & 0 & 0\\
0 & 0 & 0 & 0\\
0 & 0 & 0 & 0\\
0 & 0 & 0 & 0
\end{array}\right)\\
\left(\begin{array}{cccc}
0 & 0 & 0 & 0\\
0 & 0 & 0 & 0\\
0 & 0 & 0 & 0\\
0 & 0 & 0 & 0
\end{array}\right) & \left(\begin{array}{cccc}
0 & -1 & 0 & 0\\
1 & 0 & 0 & 0\\
0 & 0 & 0 & -1\\
0 & 0 & 1 & 0
\end{array}\right)
\end{array}\right)\\
 & =\left(\begin{array}{cccccccc}
0 & -1 & 0 & 0 & 0 & 0 & 0 & 0\\
1 & 0 & 0 & 0 & 0 & 0 & 0 & 0\\
0 & 0 & 0 & -1 & 0 & 0 & 0 & 0\\
0 & 0 & 1 & 0 & 0 & 0 & 0 & 0\\
0 & 0 & 0 & 0 & 0 & -1 & 0 & 0\\
0 & 0 & 0 & 0 & 1 & 0 & 0 & 0\\
0 & 0 & 0 & 0 & 0 & 0 & 0 & -1\\
0 & 0 & 0 & 0 & 0 & 0 & 1 & 0
\end{array}\right)=\left(\begin{array}{cccc}
1 & 0 & 0 & 0\\
0 & 1 & 0 & 0\\
0 & 0 & 1 & 0\\
0 & 0 & 0 & 1
\end{array}\right)\otimes\left(\begin{array}{cc}
0 & -1\\
1 & 0
\end{array}\right)
\end{align*}
 is in $4\times4$ complex notation
\[
\left(\begin{array}{cccc}
\ii & 0 & 0 & 0\\
0 & \ii & 0 & 0\\
0 & 0 & \ii & 0\\
0 & 0 & 0 & \ii
\end{array}\right)=\left(\begin{array}{cc}
\left(\begin{array}{cc}
\ii & 0\\
0 & \ii
\end{array}\right) & \left(\begin{array}{cc}
0 & 0\\
0 & 0
\end{array}\right)\\
\left(\begin{array}{cc}
0 & 0\\
0 & 0
\end{array}\right) & \left(\begin{array}{cc}
\ii & 0\\
0 & \ii
\end{array}\right)
\end{array}\right)=\ii\left(\begin{array}{cccc}
1 & 0 & 0 & 0\\
0 & 1 & 0 & 0\\
0 & 0 & 1 & 0\\
0 & 0 & 0 & 1
\end{array}\right),
\]
and so it just looks like multiplication by $\ii$. Thus, in this
framework, a more direct way of asking for $U(1)$-invariance would
be asking for `complex linearity'. The drawback is that
\begin{align*}
Q^{\prime}\left(0,e_{2}\right)= & \left(\begin{array}{cccc}
0 & 0 & -1 & 0\\
0 & 0 & 0 & 1\\
1 & 0 & 0 & 0\\
0 & -1 & 0 & 0
\end{array}\right)=\left(\begin{array}{cc}
0 & -\conj\\
\conj & 0
\end{array}\right),\\
Q^{\prime}\left(0,e_{3}\right)= & \left(\begin{array}{cccc}
0 & 0 & 0 & -1\\
0 & 0 & -1 & 0\\
0 & 1 & 0 & 0\\
1 & 0 & 0 & 0
\end{array}\right)=\left(\begin{array}{cc}
0 & -\ii\conj\\
\ii\conj & 0
\end{array}\right),
\end{align*}
are not commuting with $\ii$. Note that $\conj$ and $\ii\conj$
are both selfadjoint, whereas multiplication by $\ii$ is skew-selfadjoint.
\end{rem}

We choose not to single out any particular normalised $e\in\mathbb{R}^{3}$
in the following. In consequence, we have
\[
\partial_{0}+\left(\begin{array}{cc}
0 & Q^{\prime}\left(0,e\right)Q\left(0,\nabla\right)\\
Q^{\prime}\left(0,e\right)Q\left(0,\nabla\right) & 0
\end{array}\right)
\]
as our standard form of the \textbf{Dirac operator} in the case of\textbf{
mass zero}. The actual novelty, if quaternions had not been largely
dismissed from the development of the physics of electromagnetic waves,
that Dirac brought into the discussion was to include an extra term
to factorise $\Delta-m^{2}$ rather than just $\Delta$. It has been
well known that block operator matrices of the form $A=\left(\begin{array}{cc}
0 & -C^{*}\\
C & 0
\end{array}\right)$ are, by the matching block form operator 
\[
\left(\begin{array}{cc}
1 & 0\\
0 & -1
\end{array}\right),
\]
turned\footnote{Actually, this argument works for general $A=\left(\begin{array}{cc}
0 & D\\
C & 0
\end{array}\right)$, but we focus on the case $A$ skew-selfadjoint, since this guarantees
that solutions of 
\[
\left(\partial_{0}+A\right)u=f
\]
satisfies
\[
t\mapsto\frac{1}{2}\left|u\left(t\right)\right|^{2}\text{ constant on }\mathbb{R}_{>0},
\]
which characterises wave phenomena. The transformation of multiplying
by $\left(\begin{array}{cc}
1 & 0\\
0 & -1
\end{array}\right)$ is sometimes referred to as ``time-reversal'', since its effect
may formally be considered as resulting from reversal of time, i.e.
the variable change $t\mapsto-t$.} into $-A$:
\begin{equation}
\left(\begin{array}{cc}
1 & 0\\
0 & -1
\end{array}\right)\left(\begin{array}{cc}
0 & -C^{*}\\
C & 0
\end{array}\right)\left(\begin{array}{cc}
1 & 0\\
0 & -1
\end{array}\right)=-\left(\begin{array}{cc}
0 & -C^{*}\\
C & 0
\end{array}\right).\label{eq:t-}
\end{equation}
In the present case we have $C=-C^{*}=Q^{\prime}\left(0,e\right)Q\left(0,\nabla\right)$.
To maintain skew-selfadjointness of the sum we would need to consider
something like
\begin{align*}
 & \left(\begin{array}{cc}
mQ^{\prime}\left(0,e\right) & Q^{\prime}\left(0,e\right)Q\left(0,\nabla\right)\\
Q^{\prime}\left(0,e\right)Q\left(0,\nabla\right) & -mQ^{\prime}\left(0,e\right)
\end{array}\right)\\
 & =\left(\begin{array}{cc}
0 & Q^{\prime}\left(0,e\right)Q\left(0,\nabla\right)\\
Q^{\prime}\left(0,e\right)Q\left(0,\nabla\right) & 0
\end{array}\right)+m\left(\begin{array}{cc}
Q^{\prime}\left(0,e\right) & 0\\
0 & -Q^{\prime}\left(0,e\right)
\end{array}\right),
\end{align*}
which for $e=e_{1}$ is indeed Dirac's suggestion, see also \prettyref{rem:Ue1-inv}.
Next, we confirm that this is indeed a factorisation of the Klein--Gordon
operator.
\begin{thm}
Let $m\in\mathbb{R}$, $e\in\mathbb{R}^{3}$ with $|e|=1$. Then
\begin{multline*}
\left(\partial_{0}-\left(\begin{array}{cc}
mQ^{\prime}\left(0,e\right) & Q^{\prime}\left(0,e\right)Q\left(0,\nabla\right)\\
Q^{\prime}\left(0,e\right)Q\left(0,\nabla\right) & -mQ^{\prime}\left(0,e\right)
\end{array}\right)\right)\left(\partial_{0}+\left(\begin{array}{cc}
mQ^{\prime}\left(0,e\right) & Q^{\prime}\left(0,e\right)Q\left(0,\nabla\right)\\
Q^{\prime}\left(0,e\right)Q\left(0,\nabla\right) & -mQ^{\prime}\left(0,e\right)
\end{array}\right)\right)\\
=\left(\partial^{2}_{0}-\Delta+m^{2}\right)1_{\mathbb{R}^{8}}.
\end{multline*}
\end{thm}

\begin{proof}
Using that 
\begin{multline*}
\left(\begin{array}{cc}
0 & Q^{\prime}\left(0,e\right)Q\left(0,\nabla\right)\\
Q^{\prime}\left(0,e\right)Q\left(0,\nabla\right) & 0
\end{array}\right)\left(\begin{array}{cc}
Q^{\prime}\left(0,e\right) & 0\\
0 & -Q^{\prime}\left(0,e\right)
\end{array}\right)\\
=\left(\begin{array}{cc}
0 & -Q^{\prime}\left(0,e\right)^{2}Q\left(0,\nabla\right)\\
Q^{\prime}\left(0,e\right)^{2}Q\left(0,\nabla\right) & 0
\end{array}\right)
\end{multline*}
and 
\begin{multline*}
\left(\begin{array}{cc}
Q^{\prime}\left(0,e\right) & 0\\
0 & -Q^{\prime}\left(0,e\right)
\end{array}\right)\left(\begin{array}{cc}
0 & Q^{\prime}\left(0,e\right)Q\left(0,\nabla\right)\\
Q^{\prime}\left(0,e\right)Q\left(0,\nabla\right) & 0
\end{array}\right)\\
=\left(\begin{array}{cc}
0 & Q^{\prime}\left(0,e\right)^{2}Q\left(0,\nabla\right)\\
-Q^{\prime}\left(0,e\right)^{2}Q\left(0,\nabla\right) & 0
\end{array}\right)
\end{multline*}
we find that
\begin{align*}
 & \left(\left(\begin{array}{cc}
0 & Q^{\prime}\left(0,e\right)Q\left(0,\nabla\right)\\
Q^{\prime}\left(0,e\right)Q\left(0,\nabla\right) & 0
\end{array}\right)+m\left(\begin{array}{cc}
Q^{\prime}\left(0,e\right) & 0\\
0 & -Q^{\prime}\left(0,e\right)
\end{array}\right)\right)^{2}\\
 & =\left(\begin{array}{cc}
0 & Q^{\prime}\left(0,e\right)Q\left(0,\nabla\right)\\
Q^{\prime}\left(0,e\right)Q\left(0,\nabla\right) & 0
\end{array}\right)^{2}+m^{2}\left(\begin{array}{cc}
Q^{\prime}\left(0,e\right) & 0\\
0 & -Q^{\prime}\left(0,e\right)
\end{array}\right)^{2}=\left(\Delta-m^{2}\right)1_{\mathbb{R}^{8}},\\
\end{align*}
where we have used $Q(0,\nabla)^{2}=-\Delta1_{\R^{4}}$ and $Q'(0,e)^{2}=-1_{\R^{4}}$
(see \prettyref{prop:commute_Q_Q'} and \prettyref{lem:Hcompurules}).
With this we get the assertion.
\end{proof}

\begin{rem}
\begin{enumerate}[(a)]

\item It is worth noting that the two factors in the above factorisation
are unitarily congruent. For this it suffices to see that, for $f\in\mathbb{R}^{3}$
normalised and orthogonal to $e$,
\begin{align*}
\left(\begin{array}{cc}
Q^{\prime}\left(0,f\right) & 0\\
0 & Q^{\prime}\left(0,f\right)
\end{array}\right)\left(\begin{array}{cc}
mQ^{\prime}\left(0,e\right) & Q^{\prime}\left(0,e\right)Q\left(0,\nabla\right)\\
Q^{\prime}\left(0,e\right)Q\left(0,\nabla\right) & -mQ^{\prime}\left(0,e\right)
\end{array}\right)\left(\begin{array}{cc}
Q^{\prime}\left(0,f\right) & 0\\
0 & Q^{\prime}\left(0,f\right)
\end{array}\right)^{*}=\\
=\left(\begin{array}{cc}
mQ^{\prime}\left(0,f\right)Q^{\prime}\left(0,e\right) & Q^{\prime}\left(0,f\right)Q^{\prime}\left(0,e\right)Q\left(0,\nabla\right)\\
Q^{\prime}\left(0,f\right)Q^{\prime}\left(0,e\right)Q\left(0,\nabla\right) & -mQ^{\prime}\left(0,f\right)Q^{\prime}\left(0,e\right)
\end{array}\right)\left(\begin{array}{cc}
-Q^{\prime}\left(0,f\right) & 0\\
0 & -Q^{\prime}\left(0,f\right)
\end{array}\right),\\
-\left(\begin{array}{cc}
mQ^{\prime}\left(0,f\right)Q^{\prime}\left(0,e\right)Q^{\prime}\left(0,f\right) & Q^{\prime}\left(0,f\right)Q^{\prime}\left(0,e\right)Q^{\prime}\left(0,f\right)Q\left(0,\nabla\right)\\
Q^{\prime}\left(0,f\right)Q^{\prime}\left(0,e\right)Q^{\prime}\left(0,f\right)Q\left(0,\nabla\right) & -mQ^{\prime}\left(0,f\right)Q^{\prime}\left(0,e\right)Q^{\prime}\left(0,f\right)
\end{array}\right)\\
=-\left(\begin{array}{cc}
mQ^{\prime}\left(0,e\right) & Q^{\prime}\left(0,e\right)Q\left(0,\nabla\right)\\
Q^{\prime}\left(0,e\right)Q\left(0,\nabla\right) & -mQ^{\prime}\left(0,e\right)
\end{array}\right).
\end{align*}
 Here we have used $Q'(0,f)Q'(0,e)=-Q'(0,e)Q'(0,f)$ and $Q'(0,f)^{2}=-1$,
see \prettyref{lem:Hcompurules} and \prettyref{prop:commute_Q_Q'}.

\item By lifting\footnote{Strictly speaking this means that we consider the Kronecker product
\[
\frac{1}{\sqrt{2}}\left(\begin{array}{cc}
1 & 1\\
1 & -1
\end{array}\right)\otimes1_{4\times4},
\]
but for simplicity of notation, we leave this to the context and simply
write $\frac{1}{\sqrt{2}}\left(\begin{array}{cc}
1 & 1\\
1 & -1
\end{array}\right)$ instead.} the matrix

\[
\frac{1}{\sqrt{2}}\left(\begin{array}{cc}
1 & 1\\
1 & -1
\end{array}\right)
\]
in $\mathbb{R}^{2}$ to the block level of $\left(\begin{array}{cc}
mQ^{\prime}\left(0,e\right) & Q^{\prime}\left(0,e\right)Q\left(0,\nabla\right)\\
Q^{\prime}\left(0,e\right)Q\left(0,\nabla\right) & -mQ^{\prime}\left(0,e\right)
\end{array}\right)$, we see that we can exchange the role of the blocks $mQ^{\prime}\left(0,e\right)$
and $Q^{\prime}\left(0,e\right)Q\left(0,\nabla\right)$. We have indeed
that
\begin{align*}
\frac{1}{\sqrt{2}}\left(\begin{array}{cc}
1 & 1\\
1 & -1
\end{array}\right)\left(\begin{array}{cc}
mQ^{\prime}\left(0,e\right) & Q^{\prime}\left(0,e\right)Q\left(0,\nabla\right)\\
Q^{\prime}\left(0,e\right)Q\left(0,\nabla\right) & -mQ^{\prime}\left(0,e\right)
\end{array}\right)\frac{1}{\sqrt{2}}\left(\begin{array}{cc}
1 & 1\\
1 & -1
\end{array}\right)^{\top}=\\
=\frac{1}{2}\left(\begin{array}{cc}
mQ^{\prime}\left(0,e\right)+Q^{\prime}\left(0,e\right)Q\left(0,\nabla\right) & Q^{\prime}\left(0,e\right)Q\left(0,\nabla\right)-mQ^{\prime}\left(0,e\right)\\
-Q^{\prime}\left(0,e\right)Q\left(0,\nabla\right)+mQ^{\prime}\left(0,e\right) & Q^{\prime}\left(0,e\right)Q\left(0,\nabla\right)+mQ^{\prime}\left(0,e\right)
\end{array}\right)\left(\begin{array}{cc}
1 & 1\\
1 & -1
\end{array}\right),\\
=\left(\begin{array}{cc}
Q^{\prime}\left(0,e\right)Q\left(0,\nabla\right) & mQ^{\prime}\left(0,e\right)\\
mQ^{\prime}\left(0,e\right) & -Q^{\prime}\left(0,e\right)Q\left(0,\nabla\right)
\end{array}\right).
\end{align*}
The result for the Dirac operator is 
\begin{align*}
\frac{1}{\sqrt{2}}\left(\begin{array}{cc}
1 & 1\\
1 & -1
\end{array}\right)\left(\partial_{0}+\left(\begin{array}{cc}
mQ^{\prime}\left(0,e\right) & Q^{\prime}\left(0,e\right)Q\left(0,\nabla\right)\\
Q^{\prime}\left(0,e\right)Q\left(0,\nabla\right) & -mQ^{\prime}\left(0,e\right)
\end{array}\right)\right)\frac{1}{\sqrt{2}}\left(\begin{array}{cc}
1 & 1\\
1 & -1
\end{array}\right)^{\top}=\\
=\partial_{0}+\left(\begin{array}{cc}
Q^{\prime}\left(0,e\right)Q\left(0,\nabla\right) & mQ^{\prime}\left(0,e\right)\\
mQ^{\prime}\left(0,e\right) & -Q^{\prime}\left(0,e\right)Q\left(0,\nabla\right)
\end{array}\right),\\
=\partial_{0}+\left(\begin{array}{cc}
Q^{\prime}\left(0,e\right)Q\left(0,\nabla\right) & 0\\
0 & -Q^{\prime}\left(0,e\right)Q\left(0,\nabla\right)
\end{array}\right)+\\
+m\left(\begin{array}{cc}
0 & Q^{\prime}\left(0,e\right)\\
Q^{\prime}\left(0,e\right) & 0
\end{array}\right),
\end{align*}
which for $e=e_{1}$ is what is known as the \textbf{Weyl--Dirac
operator}. For $m=0$ the two blocks decouple.
\item\label{rem:tri-Dirac} It may be worth noting at this time that
we are solidly in the realm of vector analysis. Indeed, with
\[
D\left(\nabla\right)\coloneqq\begin{pmatrix}0 & 0 & 0 & 0\\
\nabla & 0 & 0 & 0\\
0 & \nabla\times & 0 & 0\\
0 & 0 & \nabla^{\top} & 0
\end{pmatrix}
\]
 we have
\[
\partial_{0}+\left(D\left(\nabla\right)-D\left(\nabla\right)^{*}\right)
\]
as massless Dirac or extended Maxwell. The corresponding Dirac mass
term is 
\[
M_{\mathrm{D},e}\coloneqq\begin{pmatrix}0 & 0 & -e^{\top} & 0\\
0 & -e\times & 0 & e\\
e & 0 & e\times & 0\\
0 & -e^{\top} & 0 & 0
\end{pmatrix}
\]
with $\left|e\right|$ as mass. For $M_{\mathrm{D},e}$ we have a
factorization as 
\[
M_{\mathrm{D},e}=-\left(D\left(e\right)-D\left(e\right)^{*}\right)\Sigma\,,
\]
where we have
\[
\Sigma\coloneqq\begin{pmatrix}0 & 0 & 0 & -1\\
0 & 0 & -1 & 0\\
0 & 1 & 0 & 0\\
1 & 0 & 0 & 0
\end{pmatrix}
\]

If we would want to make the transition to alternating differential
forms (alternating tensors), we would write $\nabla$ for the spatial
covariant derivative and $\nabla\wedge$ for the exterior derivative,
then the spatial Dirac operator takes on the rather elegant form
\[
\left(\nabla\wedge\right)-\left(\nabla\wedge\right)^{*}+\left(a\wedge\right)-\left(a\wedge\right)^{*}*\kappa
\]
acting on the direct sum of $L^{2}$-spaces of alternating differential
forms over all degrees. Here $*$ denotes the Hodge star operator
and $\kappa$ is a sign change operator with 
\[
\kappa\omega=\left(-1\right)^{q(q-1)/2}\omega=\left(-1\right)^{\left\lfloor q/2\right\rfloor }\omega,
\]
if $\omega$ is a q-form, $q\in\left\{ 0,1,2,3\right\} $, and $a$
is a (constant) 1-form, $\left|a\right|$ is the mass parameter. 

\end{enumerate}
\end{rem}

\section{The Majorana Equation}

Majorana developed a complex coefficients Dirac type equation which
actually commuted with complex conjugation, thus allowing to decouple
real and imaginary parts of any solution. We refer to \cite{Marsch2011}
and \cite{Marsch2012} for more recent in-depth considerations associated
with the Majorana equation. Our look at the Majorana equation is still
from a rather elementary perspective.

In the light of the Weyl--Dirac operator decoupling for $m=0$, one
may interpret Majorana's contribution as providing for a mass term
which is also block diagonal and so maintains the decoupling. Replacing
the Weyl--Dirac mass term $m\left(\begin{array}{cc}
0 & Q^{\prime}\left(0,e\right)\\
Q^{\prime}\left(0,e\right) & 0
\end{array}\right)$ by a mass term of the form $m\left(\begin{array}{cc}
Q^{\prime}\left(0,f\right) & 0\\
0 & -Q^{\prime}\left(0,f\right)
\end{array}\right)$, where $f$ is a normalised vector perpendicular to $e$, leads instead
of the Weyl--Dirac operator to the block diagonal system
\[
\left(\begin{array}{cc}
\partial_{0}+Q^{\prime}\left(0,e\right)Q\left(0,\nabla\right)+mQ^{\prime}\left(0,f\right) & 0\\
0 & \partial_{0}-\left(Q^{\prime}\left(0,e\right)Q\left(0,\nabla\right)+mQ^{\prime}\left(0,f\right)\right)
\end{array}\right).
\]
We refer to this as the \textbf{Majorana--Dirac operator}. Next,
we show that the sign change in the block diagonal can be done away
with.
\begin{prop}
Let $e,f\in\R^{3}$ normalised with $e\bot f$. Then 
\begin{multline*}
\left(\begin{array}{cc}
Q^{\prime}\left(0,e\right)Q\left(0,\nabla\right)+mQ^{\prime}\left(0,f\right) & 0\\
0 & -\left(Q^{\prime}\left(0,e\right)Q\left(0,\nabla\right)+mQ^{\prime}\left(0,f\right)\right)
\end{array}\right)\\
\rightleftharpoons\left(\begin{array}{cc}
Q^{\prime}\left(0,e\right)Q\left(0,\nabla\right)+mQ^{\prime}\left(0,f\right) & 0\\
0 & Q^{\prime}\left(0,e\right)Q\left(0,\nabla\right)+mQ^{\prime}\left(0,f\right)
\end{array}\right).
\end{multline*}
\end{prop}

\begin{proof}
Indeed, choosing $Q'(0,e\times f)$ as the unitary operator, we obtain
\begin{align*}
 & Q^{\prime}\left(0,e\times f\right)\left(Q^{\prime}\left(0,e\right)Q\left(0,\nabla\right)+mQ^{\prime}\left(0,f\right)\right)Q^{\prime}\left(0,e\times f\right)^{*}\\
 & =-\left(Q^{\prime}\left(0,e\times f\right)Q^{\prime}\left(0,e\right)Q^{\prime}\left(0,e\times f\right)Q\left(0,\nabla\right)+mQ^{\prime}\left(0,e\times f\right)Q^{\prime}\left(0,f\right)Q^{\prime}\left(0,e\times f\right)\right)\\
 & =-\left(Q^{\prime}\left(0,e\right)Q\left(0,\nabla\right)+mQ^{\prime}\left(0,f\right)\right),
\end{align*}
where we again have used $Q^{\prime}\left(0,e\times f\right)Q^{\prime}\left(0,e\right)=-Q^{\prime}\left(0,e\right)Q^{\prime}\left(0,e\times f\right)$
and $Q^{\prime}\left(0,e\times f\right)^{2}=-1_{\R^{4}}$ by \prettyref{prop:commute_Q_Q'}
and \prettyref{lem:Hcompurules}. Thus, the assertion follows by choosing
$U=\left(\begin{array}{cc}
1 & 0\\
0 & Q^{\prime}\left(0,e\times f\right)
\end{array}\right)$ as the unitary transformation. 
\end{proof}

Thus, the Majorana--Dirac operator is actually unitarily congruent
to 
\[
\left(\begin{array}{cc}
\partial_{0}+Q^{\prime}\left(0,e\right)Q\left(0,\nabla\right)+mQ^{\prime}\left(0,f\right) & 0\\
0 & \partial_{0}+Q^{\prime}\left(0,e\right)Q\left(0,\nabla\right)+mQ^{\prime}\left(0,f\right)
\end{array}\right)
\]
We refer to 
\[
\partial_{0}+Q^{\prime}\left(0,e\right)Q\left(0,\nabla\right)+mQ^{\prime}\left(0,f\right)
\]
as the actual \textbf{Majorana operator}, (for $e=e_{1}$ and $f=e_{2}$
in the classical case). Also the Majorana operator provides a convenient
(and actually smaller in terms of matrix size) solution of the factorisation
problem for the Klein--Gordon operator.
\begin{prop}
\label{prop:MajoSolvesKGE}Let $m\in\mathbb{R},e,f\in\mathbb{R}^{3}$
with $|e|=|f|=1$ and $e\bot f$. Then
\[
\left(\partial_{0}+Q^{\prime}\left(0,e\right)Q\left(0,\nabla\right)+mQ^{\prime}\left(0,f\right)\right)\left(\partial_{0}-\left(Q^{\prime}\left(0,e\right)Q\left(0,\nabla\right)+mQ^{\prime}\left(0,f\right)\right)\right)=\left(\partial^{2}_{0}-\Delta+m^{2}\right)1_{\mathbb{R}^{4}}.
\]
\end{prop}

\begin{proof}
We compute
\begin{multline*}
\left(\partial_{0}+Q^{\prime}\left(0,e\right)Q\left(0,\nabla\right)+mQ^{\prime}\left(0,f\right)\right)\left(\partial_{0}-\left(Q^{\prime}\left(0,e\right)Q\left(0,\nabla\right)+mQ^{\prime}\left(0,f\right)\right)\right)\\
=\partial^{2}_{0}-\left(Q^{\prime}\left(0,e\right)Q\left(0,\nabla\right)+mQ^{\prime}\left(0,f\right)\right)^{2}
\end{multline*}
and from \prettyref{lem:Hcompurules} it follows
\begin{align*}
 & \left(Q^{\prime}\left(0,e\right)Q\left(0,\nabla\right)+mQ^{\prime}\left(0,f\right)\right)^{2}\\
 & =\left(Q^{\prime}\left(0,e\right)Q\left(0,\nabla\right)\right)^{2}+mQ^{\prime}\left(0,e\right)Q^{\prime}\left(0,f\right)Q\left(0,\nabla\right)+mQ^{\prime}\left(0,f\right)Q^{\prime}\left(0,e\right)Q\left(0,\nabla\right)+\left(mQ^{\prime}\left(0,f\right)\right)^{2}\\
 & =\Delta1_{\mathbb{R}^{4}}-m^{2}1_{\mathbb{R}^{4}},
\end{align*}
which together yield the assertion.
\end{proof}

\begin{rem}
\label{rem:Cinv}\begin{enumerate}[(a)]

\item Applying the same unitary transformation for the Weyl--Dirac
system, we obtain the unitarily congruent version
\[
\partial_{0}+\left(\begin{array}{cc}
Q^{\prime}\left(0,e\right)Q\left(0,\nabla\right) & mQ^{\prime}\left(0,f\right)\\
mQ^{\prime}\left(0,f\right) & Q^{\prime}\left(0,e\right)Q\left(0,\nabla\right)
\end{array}\right).
\]
The same transformation applied to the Dirac operator yields the less
familiar looking unitarily congruent version\footnote{Strangely, despite being unitarily congruent, this version is not
U$\left(1\right)$ invariant, since two different, non-commuting,
imaginary units have appeared; see also \prettyref{rem:Ue1-inv}.}
\[
\partial_{0}+\left(\begin{array}{cc}
mQ^{\prime}\left(0,e\right) & Q^{\prime}\left(0,f\right)Q\left(0,\nabla\right)\\
Q^{\prime}\left(0,f\right)Q\left(0,\nabla\right) & mQ^{\prime}\left(0,e\right)
\end{array}\right).
\]

\item Let us now consider the selfadjoint, unitary operator
\[
\mathfrak{C}\coloneqq\left(\begin{array}{cc}
0 & -Q^{\prime}\left(0,e\times f\right)\\
Q^{\prime}\left(0,e\times f\right) & 0
\end{array}\right).
\]
Then, by the same reason as above, we have
\begin{multline*}
Q^{\prime}\left(0,e\times f\right)Q^{\prime}\left(0,e\right)Q^{\prime}\left(0,e\times f\right)Q\left(0,\nabla\right)+mQ^{\prime}\left(0,e\times f\right)Q^{\prime}\left(0,f\right)Q^{\prime}\left(0,e\times f\right)\\
=Q^{\prime}\left(0,e\right)Q\left(0,\nabla\right)+mQ^{\prime}\left(0,f\right)
\end{multline*}
and hence, 
\begin{align*}
\mathfrak{C}\left(\begin{array}{cc}
Q^{\prime}\left(0,e\right)Q\left(0,\nabla\right)+mQ^{\prime}\left(0,f\right) & 0\\
0 & -Q^{\prime}\left(0,e\right)Q\left(0,\nabla\right)-mQ^{\prime}\left(0,f\right)
\end{array}\right)\mathfrak{C}^{*}=\\
=\left(\begin{array}{cc}
Q^{\prime}\left(0,e\right)Q\left(0,\nabla\right)+mQ^{\prime}\left(0,f\right) & 0\\
0 & -Q^{\prime}\left(0,e\right)Q\left(0,\nabla\right)-mQ^{\prime}\left(0,f\right)
\end{array}\right),
\end{align*}
that is, $\left(\begin{array}{cc}
Q^{\prime}\left(0,e\right)Q\left(0,\nabla\right)+mQ^{\prime}\left(0,f\right) & 0\\
0 & -Q^{\prime}\left(0,e\right)Q\left(0,\nabla\right)-mQ^{\prime}\left(0,f\right)
\end{array}\right)$ is invariant under the transformation with $\mathfrak{C}$. 

\item Let us consider the Weyl form of the Dirac mass term. Then
the transformation $\mathfrak{C}$ just applied leads to
\begin{align*}
\mathfrak{C}\left(\begin{array}{cc}
0 & Q^{\prime}\left(0,e\right)\\
Q^{\prime}\left(0,e\right) & 0
\end{array}\right)\mathfrak{C}^{*} & =\left(\begin{array}{cc}
0 & Q^{\prime}\left(0,e\times f\right)Q^{\prime}\left(0,e\right)Q^{\prime}\left(0,e\times f\right)\\
Q^{\prime}\left(0,e\times f\right)Q^{\prime}\left(0,e\right)Q^{\prime}\left(0,e\times f\right) & 0
\end{array}\right)\\
 & =\left(\begin{array}{cc}
0 & Q^{\prime}\left(0,e\right)\\
Q^{\prime}\left(0,e\right) & 0
\end{array}\right).
\end{align*}
In other words, also the mass term in the Weyl form of the Dirac equation
is invariant under transformation with $\mathfrak{C}$.

\item Also the Majorana case is thus reconsidered in a completely
vector analytical language. Indeed, analogously and in the same notation
to \eqref{rem:Ue1-inv}, \eqref{rem:tri-Dirac}, we have now 
\[
\partial_{0}+\left(D\left(\nabla+f\right)-D\left(\nabla+f\right)^{*}\right)
\]
as a description of the Majorana operator. Here the vector $f$ is
arbitrary as long as $\left|f\right|$ is the wanted mass.

In differential form notation this is even more elegant than the Dirac
mass term. Indeed, we have 
\[
\left(\left(\nabla+f\right)\wedge\right)-\left(\left(\nabla+f\right)\wedge\right)^{*},
\]
where $f$ is a (constant) 1-form and $\left|f\right|$ is the mass
parameter.

\end{enumerate}
\end{rem}

\section{Charge Conjugation}

In this section, we address the concept of ``charge conjugation''
for the Dirac equation/operator. For this, it is necessary to introduce
yet another term, the charge term. Let $e,f\in\mathbb{R}^{3}$ be
normalised and orthogonal to one another and recall 
\begin{equation}
\mathfrak{C}\coloneqq\left(\begin{array}{cc}
0 & -Q^{\prime}\left(0,e\times f\right)\\
Q^{\prime}\left(0,e\times f\right) & 0
\end{array}\right)\label{eq:defC}
\end{equation}
from the previous remark. The additional charge term, is another skew-selfadjoint
matrix
\[
q\left(\begin{array}{cc}
Q\left(0,A\right)+A_{0}Q^{\prime}\left(0,e\right) & 0\\
0 & -Q\left(0,A\right)+A_{0}Q^{\prime}\left(0,e\right)
\end{array}\right),
\]
where $\left(A_{0},A\right)\in\R^{4}$ is now constructed as a ``potential''
from an underlying electromagnetic field, $q\in\mathbb{R}$ denotes
the charge. This term is then added to both the Weyl--Dirac operator
and the Majorana--Dirac operator. We have already seen that both
these operators are invariant under similarity transformation with
$\mathfrak{C}$, see \prettyref{rem:Cinv}. The charge term however
behaves differently.
\begin{prop}
\label{prop:chargconj}Let $e,f\in\mathbb{R}^{3},$$|e|=|f|=1,e\bot f,(A_{0},A)\in\mathbb{R}^{4}$.
Then with $\mathfrak{C}$ from \prettyref{eq:defC} 
\begin{multline*}
\mathfrak{C}\left(\begin{array}{cc}
Q\left(0,A\right)+A_{0}Q^{\prime}\left(0,e\right) & 0\\
0 & -Q\left(0,A\right)+A_{0}Q^{\prime}\left(0,e\right)
\end{array}\right)\mathfrak{C}^{*}\\
=-\left(\begin{array}{cc}
Q\left(0,A\right)+A_{0}Q^{\prime}\left(0,e\right) & 0\\
0 & -Q\left(0,A\right)+A_{0}Q^{\prime}\left(0,e\right)
\end{array}\right).
\end{multline*}
\end{prop}

\begin{proof}
First of all we note that for matrices $C,D\in\mathbb{R}^{4\times4}$
we have
\[
\mathfrak{C}\left(\begin{array}{cc}
C & 0\\
0 & D
\end{array}\right)\mathfrak{C}^{*}=\left(\begin{array}{cc}
-Q^{\prime}\left(0,e\times f\right)DQ^{\prime}\left(0,e\times f\right) & 0\\
0 & -Q^{\prime}\left(0,e\times f\right)CQ^{\prime}\left(0,e\times f\right)
\end{array}\right).
\]
Next, by \prettyref{prop:commute_Q_Q'} and \prettyref{lem:Hcompurules},
we get
\begin{align*}
-Q^{\prime}\left(0,e\times f\right)Q\left(0,A\right)Q^{\prime}\left(0,e\times f\right) & =-Q^{\prime}\left(0,e\times f\right)Q^{\prime}\left(0,e\times f\right)Q\left(0,A\right)=Q\left(0,A\right),\\
-Q^{\prime}\left(0,e\times f\right)Q^{\prime}\left(0,e\right)Q^{\prime}\left(0,e\times f\right) & =Q^{\prime}\left(0,e\times f\right)Q^{\prime}\left(0,e\times f\right)Q^{\prime}\left(0,e\right)=-Q^{\prime}\left(0,e\right).
\end{align*}
Hence,
\begin{align*}
 & \mathfrak{C}\left(\begin{array}{cc}
Q\left(0,A\right)+A_{0}Q^{\prime}\left(0,e\right) & 0\\
0 & -Q\left(0,A\right)+A_{0}Q^{\prime}\left(0,e\right)
\end{array}\right)\mathfrak{C}^{*}\\
 & =\left(\begin{array}{cc}
-Q^{\prime}\left(0,e\times f\right)\left(-Q\left(0,A\right)+A_{0}Q^{\prime}\left(0,e\right)\right)Q^{\prime}\left(0,e\times f\right) & 0\\
0 & -Q^{\prime}\left(0,e\times f\right)\left(Q\left(0,A\right)+A_{0}Q^{\prime}\left(0,e\right)\right)Q^{\prime}\left(0,e\times f\right)
\end{array}\right)\\
 & =\left(\begin{array}{cc}
-Q\left(0,A\right)-A_{0}Q^{\prime}\left(0,e\right) & 0\\
0 & Q\left(0,A\right)-A_{0}Q^{\prime}\left(0,e\right)
\end{array}\right)\\
 & =-\left(\begin{array}{cc}
Q\left(0,A\right)+A_{0}Q^{\prime}\left(0,e\right) & 0\\
0 & -Q\left(0,A\right)+A_{0}Q^{\prime}\left(0,e\right)
\end{array}\right).\qedhere
\end{align*}
\end{proof}

\begin{rem}
If a charge term is present, then, by \prettyref{prop:chargconj},
the transformation with $\mathfrak{C}$ results in the transition
from 
\[
q\left(\begin{array}{cc}
Q\left(0,A\right)+A_{0}Q^{\prime}\left(0,e\right) & 0\\
0 & -Q\left(0,A\right)+A_{0}Q^{\prime}\left(0,e\right)
\end{array}\right)
\]
 to 
\[
-q\left(\begin{array}{cc}
Q\left(0,A\right)+A_{0}Q^{\prime}\left(0,e\right) & 0\\
0 & -Q\left(0,A\right)+A_{0}Q^{\prime}\left(0,e\right)
\end{array}\right),
\]
that is, we have an opposite charge in both the Weyl--Dirac and the
Majorana--Dirac operator case. Since in the Weyl--Dirac operator
case we only have $Q^{\prime}\left(0,e\right)$ as a coefficient in
$\mathbb{H}^{\prime}$ occurring, there is, strangely enough, another
``charge conjugation'' available namely
\[
\mathfrak{C}^{\prime}\coloneqq\left(\begin{array}{cc}
0 & -Q^{\prime}\left(0,f\right)\\
Q^{\prime}\left(0,f\right) & 0
\end{array}\right),
\]
the calculation being the same as for $\mathfrak{C}$. One could argue
that the Majorana--Dirac operator should be favoured here as there
is no ambiguity in the choice of $\mathfrak{C}$, as $e$ and $f$
are already used in its formulation. For the Weyl--Dirac equation
there the transformations $\mathfrak{C}$ and $\mathfrak{C}^{\prime}$
are possible ``charge conjugation'' operators. 
\end{rem}

\section{The \textquotedblleft Majorana Particle\textquotedblright}

There is recently increasing interest in a so-called ``Majorana particle'',
which is described as being its own ``anti-particle''. This stems
from the observation that the Majorana mass term is linked through
$\mathfrak{C}$ to the Weyl--Dirac mass term via the equalities
\begin{equation}
\left(\begin{array}{cc}
Q^{\prime}\left(0,f\right) & 0\\
0 & -Q^{\prime}\left(0,f\right)
\end{array}\right)=\left(\begin{array}{cc}
0 & Q^{\prime}\left(0,e\right)\\
Q^{\prime}\left(0,e\right) & 0
\end{array}\right)\mathfrak{C}=\mathfrak{C}\left(\begin{array}{cc}
0 & Q^{\prime}\left(0,e\right)\\
Q^{\prime}\left(0,e\right) & 0
\end{array}\right).\label{eq:Cecom}
\end{equation}
Thus the Majorana--Dirac equation can be written as
\[
\left(\partial_{0}+\left(\begin{array}{cc}
Q^{\prime}\left(0,e\right)Q\left(0,\nabla\right) & 0\\
0 & -Q^{\prime}\left(0,e\right)Q\left(0,\nabla\right)
\end{array}\right)\right)\Psi+m\left(\begin{array}{cc}
0 & Q^{\prime}\left(0,e\right)\\
Q^{\prime}\left(0,e\right) & 0
\end{array}\right)\Psi_{\mathfrak{C}}=F,
\]
where $\Psi_{\mathfrak{C}}\coloneqq\mathfrak{C}\Psi.$ This coincides
with the Weyl--Dirac equation if 
\[
\Psi=\Psi_{\mathfrak{C}},
\]
that is,
\[
\Psi\in\ker\left(1-\mathfrak{C}\right),
\]
which to happen implies a constraint on the data as the next statement
confirms.
\begin{thm}
\label{thm:majopart}Let $e,f\in\mathbb{R}^{3}$ normalised and orthogonal,
$\mathfrak{C}$ as in \prettyref{eq:defC}. Then, for smooth $\Psi$
and $F$ such that
\[
\left(\partial_{0}+\left(\begin{array}{cc}
Q^{\prime}\left(0,e\right)Q\left(0,\nabla\right) & 0\\
0 & -Q^{\prime}\left(0,e\right)Q\left(0,\nabla\right)
\end{array}\right)\right)\Psi+m\left(\begin{array}{cc}
0 & Q^{\prime}\left(0,e\right)\\
Q^{\prime}\left(0,e\right) & 0
\end{array}\right)\Psi=F,
\]
we have
\[
\Psi=\mathfrak{C}\Psi\iff F=\mathfrak{C}F.
\]
\end{thm}

\begin{proof}
We shall assume $\Psi=\mathfrak{C}\Psi$ first. Then we compute using
\prettyref{eq:Cecom} and \prettyref{rem:Cinv} (and $\mathfrak{C}=\mathfrak{C}^{*}=\mathfrak{C}^{-1}$)
to obtain
\begin{align*}
F & =\left(\partial_{0}+\left(\begin{array}{cc}
Q^{\prime}\left(0,e\right)Q\left(0,\nabla\right) & 0\\
0 & -Q^{\prime}\left(0,e\right)Q\left(0,\nabla\right)
\end{array}\right)\right)\Psi+m\left(\begin{array}{cc}
0 & Q^{\prime}\left(0,e\right)\\
Q^{\prime}\left(0,e\right) & 0
\end{array}\right)\Psi\\
 & =\left(\partial_{0}+\left(\begin{array}{cc}
Q^{\prime}\left(0,e\right)Q\left(0,\nabla\right) & 0\\
0 & -Q^{\prime}\left(0,e\right)Q\left(0,\nabla\right)
\end{array}\right)\right)\mathfrak{C}\Psi+m\left(\begin{array}{cc}
0 & Q^{\prime}\left(0,e\right)\\
Q^{\prime}\left(0,e\right) & 0
\end{array}\right)\mathfrak{C}\Psi\\
 & =\mathfrak{C}\left(\partial_{0}+\left(\begin{array}{cc}
Q^{\prime}\left(0,e\right)Q\left(0,\nabla\right) & 0\\
0 & -Q^{\prime}\left(0,e\right)Q\left(0,\nabla\right)
\end{array}\right)\right)\Psi+m\mathfrak{C}\left(\begin{array}{cc}
0 & Q^{\prime}\left(0,e\right)\\
Q^{\prime}\left(0,e\right) & 0
\end{array}\right)\Psi\\
 & =\mathfrak{C}F.
\end{align*}
On the other hand, if $\mathfrak{C}F=F,$ then by a similar computation
just carried out, we infer
\begin{multline*}
\left(\partial_{0}+\left(\begin{array}{cc}
Q^{\prime}\left(0,e\right)Q\left(0,\nabla\right) & 0\\
0 & -Q^{\prime}\left(0,e\right)Q\left(0,\nabla\right)
\end{array}\right)\right)\Psi+m\left(\begin{array}{cc}
0 & Q^{\prime}\left(0,e\right)\\
Q^{\prime}\left(0,e\right) & 0
\end{array}\right)\Psi\\
=\left(\partial_{0}+\left(\begin{array}{cc}
Q^{\prime}\left(0,e\right)Q\left(0,\nabla\right) & 0\\
0 & -Q^{\prime}\left(0,e\right)Q\left(0,\nabla\right)
\end{array}\right)\right)\mathfrak{C}\Psi+m\left(\begin{array}{cc}
0 & Q^{\prime}\left(0,e\right)\\
Q^{\prime}\left(0,e\right) & 0
\end{array}\right)\mathfrak{C}\Psi
\end{multline*}
Since the operator $\left(\begin{array}{cc}
Q^{\prime}\left(0,e\right)Q\left(0,\nabla\right) & mQ^{\prime}\left(0,e\right)\\
mQ^{\prime}\left(0,e\right) & -Q^{\prime}\left(0,e\right)Q\left(0,\nabla\right)
\end{array}\right)$ is (easily) confirmed to be skew-selfadjoint with its maximal domain
in $L^{2}(\mathbb{R}^{3};\mathbb{R}^{8})$ the theory of evolutionary
equations (see, e.g., \cite{Picard2009,PicardMcGhee2011,Primer})
or $C_{0}$-semigroups yield uniqueness of smooth solutions. Thus,
$\mathfrak{C}\Psi=\Psi$, as desired.
\end{proof}

\begin{rem}
~
\begin{enumerate}
\item The upshot of \prettyref{thm:majopart} is that looking for the ``Majorana
particle'' amounts to look for constraints on the values of $F$ rather
than additional properties of $\Psi$. Even if, as one does, only
considers initial data, this just means
\[
F=\delta\otimes\Psi_{0+}\:,
\]
in other words, initial data are just data concentrated at $t=0$,
implying that the initial data $\Psi_{0+}$ are constrained in the
same way
\[
\Psi_{0+}\in\ker\left(1-\mathfrak{C}\right).
\]
Compared to other concepts of particles, it seems rather odd to create
new particles by restricting data to an invariant subspace. \\
Following the standard procedure of calling a particle Majorana particle
(without the quotation marks), we would just be asking for particles
being described by the Majorana operator. However, so far the physicists
community has dismissed the Majorana operator for theoretical reason,
since it is incompatible with further developments of quantum theory.
This is of course not surprising, since these developments were based
on the Dirac operator.
\item It may be interesting to consider, what the development might have
been if the Majorana--Dirac mass term had been discovered first.
After all Dirac found his equation by trying to factorise the Klein--Gordon
equation and the Majorana--Dirac operator does the trick just as
well, see \prettyref{prop:MajoSolvesKGE}. In fact, starting with
smooth $\Psi_{\mu}$ and $F$ satisfying the equation given by the
Majorana--Dirac operator, that is,
\[
\left(\partial_{0}+\left(\begin{array}{cc}
Q^{\prime}\left(0,e\right)Q\left(0,\nabla\right) & 0\\
0 & -Q^{\prime}\left(0,e\right)Q\left(0,\nabla\right)
\end{array}\right)\right)\Psi_{\mu}+m\left(\begin{array}{cc}
Q^{\prime}\left(0,f\right) & 0\\
0 & -Q^{\prime}\left(0,f\right)
\end{array}\right)\Psi_{\mu}=F,
\]
and using 
\[
\left(\begin{array}{cc}
Q^{\prime}\left(0,f\right) & 0\\
0 & -Q^{\prime}\left(0,f\right)
\end{array}\right)=\left(\begin{array}{cc}
0 & Q^{\prime}\left(0,e\right)\\
Q^{\prime}\left(0,e\right) & 0
\end{array}\right)\mathfrak{C},
\]
we get
\begin{align*}
\left(\partial_{0}+\left(\begin{array}{cc}
Q^{\prime}\left(0,e\right)Q\left(0,\nabla\right) & 0\\
0 & -Q^{\prime}\left(0,e\right)Q\left(0,\nabla\right)
\end{array}\right)\right)\Psi_{\mu}+m\left(\begin{array}{cc}
0 & Q^{\prime}\left(0,e\right)\\
Q^{\prime}\left(0,e\right) & 0
\end{array}\right)\Psi_{\mu,\mathfrak{C}} & =F
\end{align*}
with $\Psi_{\mu,\mathfrak{C}}=\mathfrak{C}\Psi_{\mu}.$ Would we then
look for a ``Dirac particle'', which by definition would have to
be its own Majorana ``anti-particle'', that is, $\Psi_{\mu}=\Psi_{\mu,\mathfrak{C}}\:$?
\item Whereas many of the above considered operators have been identified
to be unitarily equivalent and, thus, implying the `same physics',
the choice for the mass term in the Majorana--Dirac or the Weyl--Dirac
perspective appears to be arbitrary from the insights following the
rationales outlined here. In fact, since $\mathfrak{C}$ leaves both
operators invariant (hence no charge conjugation), there is actually
no anti-particle in this case, anyway. In the light of the above observations
it appears more reasonable to ask, ``Which operator is the right
one?'', the Dirac operator or the Majorana alternative? How could
this be decided by measuring? 
\end{enumerate}
Whether or not it is possible to distinguish these alternatives by
measuring is unclear, yet the next section alludes to a coherent perspectives
concerning the spin, if the Majorana operator is favoured.
\end{rem}

\section{Spin 0, Spin $\frac{1}{2}$, Spin 1,... from the Perspective of the
Majorana Operator}

To begin with, we like to point out one more curiosity linked to the
Majorana operator. The Klein--Gordon operator $\partial^{2}_{0}-\Delta+m^{2}$
is frequently referred to as governing the spin $0$ particles in
the light of the Dirac operator describing spin $\frac{1}{2}$ particles.
In view of the sequence of the intricate standard constructions for
operators associated with higher spin numbers, which are all first
order systems, it appears rather odd, that there should be a second
order system describing spin $0$ particles. The problem is, there
is no Dirac type equation in $\mathbb{C}^{2}$ or $\mathbb{R}^{4}$. 

Well, in the Majorana operator perspective there is no such problem:
the Majorana operator is what should be called as describing spin
$0$ particles, that is operators of the form $\partial_{0}+\mathcal{M}$
with 
\[
\mathcal{M}=Q^{\prime}\left(0,e\right)Q\left(0,\nabla\right)+mQ^{\prime}\left(0,f\right).
\]
Compared to the Majorana--Dirac operator, the spin $\frac{1}{2}$
case, if you will, is simply of the block-diagonal form (adding another
Majorana component yielding an increase in spin by $\frac{1}{2}$)
\[
\partial_{0}+\left(\begin{array}{cc}
\mathcal{M} & 0\\
0 & -\mathcal{M}
\end{array}\right).
\]
This would lead to a completely different conjecture for spin 1 particles,
namely the block-diagonal system
\[
\partial_{0}+\left(\begin{array}{ccc}
\mathcal{M} & 0 & 0\\
0 & -\mathcal{M} & 0\\
0 & 0 & \mathcal{M}
\end{array}\right).
\]
Since, as we have seen, we may even do away with the sign change in
the blocks, we could alternatively consider, with the abbreviation
\[
\diag\left(\mathcal{M},\ldots,\mathcal{M}\right)\coloneqq\left(\begin{array}{cccc}
\mathcal{M} & 0 & \cdots & 0\\
0 & \ddots & \ddots & \vdots\\
\vdots & \ddots & \ddots & 0\\
0 & \cdots & 0 & \mathcal{M}
\end{array}\right),
\]
 the block diagonal operator matrix
\[
\partial_{0}+\diag\left(\mathcal{M},\ldots,\mathcal{M}\right)=\diag\left(\partial_{0}+\mathcal{M},\ldots,\partial_{0}+\mathcal{M}\right),
\]
with repeated Majorana operators in the diagonal. This would make
the spin number just a counter of the number of blocks in the diagonal
(as it also turns out to be in the more involved standard construction
for higher spin numbers):
\[
\ell=\frac{b-1}{2}\in\frac{1}{2}\left[\mathbb{N}\right].
\]
This could be considered a cleanup of Majorana's original idea to
utilise his alternative factorisation of the Klein--Gordon operator
to have a simple mechanism for higher spin particles. Since he based
his proposal on the spin $\frac{1}{2}$ case, he developed this idea
for spin $\frac{2k+1}{2},k\in\mathbb{N}$. His attempt to consider
all these spin numbers in one step was likely misguided and lead to
his theory to become an easier target for dismissal. With the above
refinement of his idea to all spin numbers, we have given Majorana's
idea about higher spin the proper completion they deserve.


\begin{thebibliography}{10}
\bibitem{Bialynicki-Birula1996} I.~Bialynicki-Birula. \newblock
V Photon Wave Function. \newblock In E.~Wolf, editor, {\em Progress
in Optics}, volume~36, pages 245--294. Elsevier, 1996.

\bibitem{CabibboFerrari1962Monopoles} N.~Cabibbo and E.~Ferrari.
\newblock Quantum Electrodynamics with Dirac Monopoles. \newblock
{\em Il Nuovo Cimento}, 23:1147--1154, 1962.

\bibitem{Dirac1928QTElectron} P.~A.~M. Dirac. \newblock The Quantum
Theory of the Electron. \newblock {\em Proceedings of the Royal
Society of London. Series A}, 117(778):610--624, 1928.

\bibitem{Dirac1928QuantumTheoryElectronPartII} P.~A.~M. Dirac.
\newblock The Quantum Theory of the Electron. part ii. \newblock
{\em Proceedings of the Royal Society of London. Series A}, 118(779):351--361,
1928.

\bibitem{Dirac1931Monopole} P.~A.~M. Dirac. \newblock Quantised
singularities in the electromagnetic field. \newblock {\em Proceedings
of the Royal Society of London. Series A}, 133(821):60--72, 1931.

\bibitem{FeshbachVillars1958ElementaryRelativistic} Herman Feshbach
and Felix Villars. \newblock Elementary Relativistic Wave Mechanics
of Spin 0 and Spin 1/2 Particles. \newblock {\em Reviews of Modern
Physics}, 30(1):24--45, 1958.

\bibitem{IwanenkoLandau1928} D.~Iwanenko and L.~Landau. \newblock
Zur Theorie des magnetischen Elektrons. i. \newblock {\em Zeitschrift
f{ü}r Physik}, 48:340--348, 1928.

\bibitem{Kaehler1961DiracGleichung} Erich K{ä}hler. \newblock
Die Dirac-Gleichung. \newblock {\em Abhandlungen der Deutschen
Akademie der Wissenschaften zu Berlin, Klasse f{ü}r Mathematik,
Physik und Technik}, (1):1--38, 1961.

\bibitem{Kaehler1962InnererDifferentialkalkuelHamburg} Erich K{ä}hler.
\newblock Der innere Differentialkalk{ü}l. \newblock {\em Abhandlungen
aus dem Mathematischen Seminar der Universit{ä}t Hamburg}, 25:192--205,
1962.

\bibitem{Kemmer1939ParticleAspectMesonTheory} Nicholas Kemmer. \newblock
The Particle Aspect of Meson Theory. \newblock {\em Proceedings
of the Royal Society of London. Series A}, 173(952):91--116, 1939.

\bibitem{Majorana1937TeoriaSimmetrica} Ettore Majorana. \newblock
Teoria Simmetrica dell'Elettrone e del Positrone. \newblock {\em
Il Nuovo Cimento}, 14:171--184, 1937.

\bibitem{Marsch2011} E.~Marsch. \newblock The Two-Component Majorana
Equation - \emph{novel derivations and known symmetries}. \newblock
{\em Journal of Modern Physics}, 2(10):1109--1114, 2011.

\bibitem{Marsch2012} E.~Marsch. \newblock On the Majorana Equation:
Relations between its Complex Two-Component and Real Four-Component
Eigenfunctions. \newblock {\em International Scholarly Research
Notices}, 2012:1--17, 2012.

\bibitem{Pi84} R.~Picard. \newblock On the Low Frequency Asymptotics
in Electromagnetic Theory. \newblock {\em J. Reine Angew. Math.},
354:50--73, 1984.

\bibitem{Picard2009} R.~Picard. \newblock A Structural Observation
for Linear Material Laws in Classical Mathematical Physics. \newblock
{\em Math. Methods Appl. Sci.}, 32(14):1768--1803, 2009.

\bibitem{PicardMcGhee2011} R.~Picard and D.~McGhee. \newblock
{\em Partial Differential Equations},\newblock A Unified Hilbert
Space Approach. Volume~55 of {\em De Gruyter Expositions in Mathematics}.
\newblock Walter de Gruyter GmbH \& Co. KG, Berlin, 2011. 

\bibitem{Primer} R.~Picard, D.~McGhee, S.~Trostorff, and M.~Waurick.
\newblock {\em A Primer for a Secret Shortcut to {PDE}s of Mathematical
Physics}. \newblock Frontiers in Mathematics. Birkhäuser/Springer,
Cham, {[}2020{]} ©2020.

\bibitem{PT2024} R.~Picard and S.~Trostorff. \newblock M-Accretive
Realisations of Skew-Symmetric Operators. \newblock {\em J. Operator
Theory}, 92(1):3--24, 2024.

\bibitem{PTW2017} R.~Picard, S.~Trostorff, and M.~Waurick. \newblock
On a Connection between the {M}axwell System, the Extended {M}axwell
System, the {D}irac Operator and Gravito-Electromagnetism. \newblock
{\em Math. Methods Appl. Sci.}, 40(2):415--434, 2017.

\bibitem{Silberstein1907} L.~Silberstein. \newblock Elektromagnetische
Grundgleichungen in bivektorieller Behandlung. \newblock {\em Annalen
der Physik}, 327(3):579--586, 1907.

\bibitem{TV07} M.~Taskinen and S.~V{ä}nsk{ä}. \newblock Current
and Charge Integral Equation Formulations and {Picard}'s Extended
{Maxwell} System. \newblock {\em IEEE Trans. Antennas Propag.},
55(12):3495--3503, 2007.

\end{thebibliography}
\end{document}